\documentclass{article}

\usepackage{cmbright}
\usepackage{amssymb,amsmath,amsthm}
\usepackage{mathrsfs}
\usepackage{epsfig}
\usepackage{color}

\def\B{\mathscr B}
\def\C{\mathbb C}
\def\D{\mathscr D}
\def\d{\mathrm d}
\def\F{\mathscr F}
\def\G{\mathcal G}
\def\H{\mathcal H}
\def\K{\mathcal K}
\def\M{\mathsf M}
\def\N{\mathbb N}
\def\O{\mathcal O}
\def\Oa{{\mathcal O}_{\rm as}}
\def\P{\mathcal P}
\def\R{\mathbb R}
\def\S{\mathscr S}
\def\T{\mathbb T}
\def\U{\mathscr U}
\def\Z{\mathbb Z}
\def\Hrond{\mathscr H}
\def\Pv{\mathrm{Pv}}
\def\Fc{\F_{\rm c}}

\def\dom{\mathcal D}
\def\Rem{\mathsf{Rem}}
\def\lone{\mathsf{L}^{\:\!\!1}}

\def\ltwo{\mathsf{L}^{\:\!\!2}}
\def\linf{\mathsf{L}^{\:\!\!\infty}}
\def\e{\mathop{\mathrm{e}}\nolimits}
\DeclareMathOperator*{\im}{Im}
\DeclareMathOperator*{\re}{Re}
\DeclareMathOperator*{\slim}{s\;\!-lim\;\!}
\DeclareMathOperator*{\ulim}{u\;\!-lim\;\!}

\newtheorem{Theorem}{Theorem}[section]
\newtheorem{Remark}[Theorem]{Remark}
\newtheorem{Lemma}[Theorem]{Lemma}

\newtheorem{Corollary}[Theorem]{Corollary}
\newtheorem{Proposition}[Theorem]{Proposition}

\begin{document}


\title{Spectral and scattering properties at thresholds for the Laplacian in a half-space
with a periodic boundary condition
}

\author{S. Richard$^1$\footnote{Supported by JSPS Grant-in-Aid for Young Scientists A no 26707005.}~~and R. Tiedra de
Aldecoa$^2$\footnote{Supported by the Chilean Fondecyt Grant 1130168 and by the
Iniciativa Cientifica Milenio ICM RC120002 ``Mathematical Physics'' from the Chilean
Ministry of Economy.}}

\date{\small}
\maketitle \vspace{-1cm}

\begin{quote}
\emph{
\begin{itemize}
\item[$^1$] Graduate school of mathematics, Nagoya University,
Chikusa-ku, Nagoya 464-8602, Japan; On leave of absence from
Universit\'e de Lyon; Universit\'e
Lyon 1; CNRS, UMR5208, Institut Camille Jordan,
43 blvd du 11 novembre 1918, F-69622
Villeurbanne-Cedex, France
\item[$^2$] Facultad de Matem\'aticas, Pontificia Universidad Cat\'olica de Chile,\\
Av. Vicu\~na Mackenna 4860, Santiago, Chile
\item[] \emph{E-mails:} richard@math.univ-lyon1.fr, rtiedra@mat.puc.cl
\end{itemize}
}
\end{quote}


\begin{abstract}
For the scattering system given by the Laplacian in a half-space with a periodic
boundary condition, we derive resolvent expansions at embedded thresholds, we prove
the continuity of the scattering matrix, and we establish new formulas for the wave
operators.
\end{abstract}


\smallskip



\section{Introduction}\label{Intro}
\setcounter{equation}{0}

We present in this paper new results for the Laplacian in a half-space subject to a
periodic boundary condition, as introduced and described by R. L. Frank and R. G.
Shterenberg in \cite{Fra03,Fra06,FS04}. We derive resolvent expansions at embedded
thresholds (which occur in an infinite number after a Floquet decomposition), we prove
the continuity of the scattering matrix at thresholds, and we establish new
representation formulas for the wave operators. These results belong to the
intersection of two active research topics in spectral and scattering theory. On one
hand, resolvent expansions at thresholds (which have a long history, but which have
been more systematically
developed since the seminal paper of A. Jensen and G. Nenciu \cite{JN01}, see also \cite{ES04,IJ13,JN04,RT14}).
On the second hand, representation formulas for the wave operators and their application to the proof of
index theorems in scattering theory (see \cite{BSB12,IR12,KR08,KR12,RT10,RT13_4,SB13}
and references therein). These results also furnish a new contribution to the very short
list of papers devoted to the subtle, and still poorly understood, topic of spectral and
scattering theory at embedded thresholds
(to our knowledge only the references \cite{BDG00,BDM06,CGDG04,Dur03,GJY04,RT14}
deal specifically with this issue).

Before giving a more precise description of our results, we recall the definition and
some of the properties (established in \cite{Fra03,Fra06,FS04}) of the model we
consider. The model consists in a scattering system $\{H^0,H^V\}$, where $H^V$ (the
perturbed operator) is the Laplacian on the half-space $\R\times\R_+$ subject to a
boundary condition on $\R\times\{0\}$ given in terms of a $2\pi$-periodic function
$V:\R\to\R$, and where $H^0$ (the unperturbed operator) is the Neumann Laplacian on
$\R\times\R_+$. An application of a Bloch-Floquet-Gelfand transform in the periodic
variable shows that the pair $\{H^0,H^V\}$ is unitarily equivalent to a family of
self-adjoint operators $\{H^0_k,H^V_k\}_{k\in[-1/2,1/2]}$ acting in the Hilbert space
$\ltwo\big((-\pi,\pi)\times\R_+\big)$. The operators $H^0_k$ have purely absolutely
continuous spectrum, whereas the operators $H^V_k$ have no singular continuous
spectrum but can have discrete spectrum (with only possible accumulation point at
$+\infty$). Under suitable conditions on $V$, it is known that the wave operators
$W_{k,\pm}:=W_\pm(H^0_k,H^V_k)$ exist and are complete, and that the full wave
operators $W_\pm:=W_\pm(H^0,H^V)$ exist, but may be not complete. The states belonging
to the cokernel of $W_\pm$ are interpreted as \emph{surface states}; that is, states
which propagate along the boundary $\R\times\{0\}$.

The completeness of the wave operators $W_{k,\pm}$ and the intertwining property imply
that the scattering operator $S_k:=W_{k,+}^*W_{k,-}$ is unitary and decomposable in
the spectral representation of $H^0_k$. However, since the spectral multiplicity of
$H^0_k$ is piecewise constant with a jump at each point of the threshold set
$$
\tau_k:=\big\{\lambda_{k,n}:=(n+k)^2\mid n\in\Z\big\},
$$
the scattering matrix $S_k(\lambda)$ can only be defined for $\lambda\notin\tau_k$.
Therefore, the continuity of $S_k(\lambda)$ in $\lambda$ can only be proved in a
suitable sense. By introducing channels corresponding to the transverse modes on the
interval $(-\pi,\pi)$, we show that $S_k(\lambda)$ is continuous at the thresholds if
the channels we consider are already open, and that $S_k(\lambda)$ has a limit from
the right at the thresholds if a channel precisely opens at these thresholds (see
Proposition \ref{propcont} for a more precise statement). Also, we give explicit
formulas for $S_k(\lambda)$ at thresholds. To our knowledge, this type of results has
never been obtained before except in \cite{RT14}, in the context of quantum
waveguides. Our proof of the continuity properties relies on a stationary
representation for $S_k(\lambda)$ and on resolvent expansions for $H^V_k$ at embedded
thresholds. The resolvent expansions are proved in Proposition \ref{Prop_Asymp} under
the single assumption that $V\in\linf(\R;\R)$. Information about the localization of
the possible eigenvalues of $H^V_k$ is also given in Section \ref{sec_spectral}.

Section \ref{sec5} is devoted to the derivation of representation formulas for the
wave operators $W_{k,\pm}$. The main result of the section are formulas
$$
W_{k,-}-1=\big(1\otimes R(A_+)\big)(S_k-1)+\Rem
\quad\hbox{and}\quad
W_{k,+}-1=\big(1-1\otimes R(A_+)\big)(S_k^*-1)+\Rem,
$$
where $R$ is the function given by
$R(x):=\frac12\big(1+\tanh(\pi x)+i\cosh(\pi x)^{-1}\big)$, $A_+$ is the generator of
dilations in $\R_+$, and $\Rem$ is a remainder term which is small in a suitable sense
(see Corollary \ref{cor_wave_k}). This type of formulas has recently been derived for
various scattering systems and is at the root of a topological approach of Levinson's
theorem (see \cite{KR12} for more explanations on this approach). Finally, collecting
the previous identities for all values $k$, we obtain similar representations formulas
for the full wave operators $W_\pm(H^0,H^V)$ (see Corollary \ref{cor_full_wave}).

The content of this paper stops here and corresponds to the analytical part of a
larger research project. As a motivation for further studies, we briefly sketch the sequel of the project here.
Under some stronger assumption on $V$, for instance if $V$ is a trigonometric polynomial, we expect the remainder term $\Rem$ to be a compact operator.
In such a case, by using appropriate techniques of K-theory and $C^*$-algebras,
one could relate the orthogonal projection on the bound states of $H^V_k$ to the scattering operator $S_k$ plus some correction terms due to threshold effects
(see for example \cite[Sec.~3]{RT10} for a presentation of the algebraic techniques
in a much simpler setting). Then, using direct integrals to collect
the results for all values of $k$, one would automatically obtain
a relation between the orthogonal projection
on the surface states of $H^V$ and operators involved in the scattering process.
This relation would be of a topological nature,
it would have an interpretation in the general context of bulk-edge correspondence,
and it would be completely new for such a continuous model.
For discrete models, related results have been obtained in \cite{Cha99}
for ergodic operators and in \cite{SB13} for deterministic operators.\\

\noindent
{\bf Acknowledgements.} The first author thanks H. Bustos and D. Parra for useful
discussions at a preliminary stage of this work.

\section{Laplacian in a half-space}\label{Sec_half}
\setcounter{equation}{0}

In this section, we recall the basic properties of the model we consider, which
consists in a Laplacian on the half-space $\R\times\R_+$, with $\R_+:=(0,\infty)$,
subject to a periodic boundary condition on $\R\times\{0\}$. Most of the material we
present here is borrowed from the papers \cite{Fra03,Fra06} to which we refer for
further information.

We choose a $2\pi$-periodic function $V\in\linf(\R;\R)$, and for each non-empty open
set $\Omega\subset\R^n$, $n\in\N^*$, and each $m\in\N$, we denote by $\H^m(\Omega)$
the usual Sobolev space of order $m$ on $\Omega$. Then, we consider the sesquilinear
form $h^V:\H^1(\R\times\R_+)\times\H^1(\R\times\R_+)\to\C$ given by
\begin{align*}
h^V(\varphi,\psi)
&:=\int_{\R\times\R_+}
\big\{\overline{(\partial_1\varphi)(x_1,x_2)}\;\!(\partial_1\psi)(x_1,x_2)
+\overline{(\partial_2\varphi)(x_1,x_2)}\;\!(\partial_2\psi)(x_1,x_2)\big\}\,
\d x_1\d x_2\\
&\quad+\int_\R V(x_1)\;\!\overline{\varphi(x_1,0)}\;\!\psi(x_1,0)\,\d x_1,
\end{align*}
where the last integral is well defined thanks to the boundary trace imbedding theorem
\cite[Thm.~5.36]{AF03}. This sesquilinear form is lower semibounded and closed, and
therefore induces in $\ltwo(\R\times\R_+)$ a lower semibounded self-adjoint operator
$H^V$ with domain $\dom(H^V)$ satisfying the equation
$$
\big\langle H^V\varphi,\psi\big\rangle_{\ltwo(\R\times\R_+)}
=h^V(\varphi,\psi),\quad\varphi\in\dom(H^V)\subset\H^1(\R\times\R_+),
~\psi\in\H^1(\R\times\R_+).
$$
In the case $V\equiv0$, the operator $H^0$ is the Neumann Laplacian on $\R\times\R_+$.

\subsection{Direct integral decomposition of $\boldsymbol{H^V}$}\label{section_direct}

Let $\S(\R^2)$ be the Schwartz space on $\R^2$ and
$
\S(\R\times\R_+)
:=\big\{\varphi\mid\varphi=\psi|_{\R\times\R_+}\hbox{ for some }\psi\in\S(\R^2)\big\}.
$
Let $\T:=(-\pi,\pi)$, set $\Pi:=\T\times\R_+$, let $\widetilde C^\infty(\Pi)$
be the set of functions in $C^\infty(\Pi)$ which can be extended $2\pi$-periodically to
functions in $C^\infty(\R\times\R_+)$, and for each $m \in \N$ let $\widetilde\H^m(\Pi)$
be the closure of $\widetilde C^\infty(\Pi)\cap \H^m(\Pi)$ in $\H^m(\Pi)$.
Then, the Gelfand transform
$\G:\S(\R\times\R_+)\to\int_{[-1/2,1/2]}^\oplus\ltwo(\Pi)\,\d k$ given by
\cite[Sec.~2.2]{Fra03}
$$
(\G\varphi)(k,\theta,x_2)
:=\sum_{n\in\Z}\e^{-ik(\theta+2\pi n)}\varphi(\theta+2\pi n,x_2),
\quad\varphi\in\S(\R\times\R_+),~k\in[-1/2,1/2],~(\theta,x_2)\in\Pi,
$$
extends to a unitary operator
$\G:\ltwo(\R\times\R_+)\to\int_{[-1/2,1/2]}^\oplus\ltwo(\Pi)\,\d k$. Moreover, one has
$$
\G H^V\G^{-1}=\int_{[-1/2,1/2]}^\oplus H^V_k\,\d k,
$$
with $H^V_k$ the lower semibounded self-adjoint operator in $\ltwo(\Pi)$ associated
with the lower semibounded and closed sesquilinear form
$
h^V_k:\widetilde\H^1(\Pi) \times \widetilde\H^1(\Pi)\to\C
$
given by
\begin{align*}
h^V_k(\varphi,\psi)
&=\int_{\Pi}\big\{\overline{\big((-i\partial_1+k)\varphi)\big)(\theta,x_2)}\;\!
\big((-i\partial_1+k)\psi)\big)(\theta,x_2)
+\overline{(\partial_2\varphi)(\theta,x_2)}\;\!(\partial_2\psi)(\theta,x_2)\big\}
\,\d \theta\;\!\d x_2\\
&\quad+\int_\T V(\theta)\;\!\overline{\varphi(\theta,0)}\;\!\psi(\theta,0)\,\d \theta.
\end{align*}

In the case $V\equiv0$, the operator $H^0_k$ reduces to
\begin{equation}\label{defH0}
H^0_k=(P+k)^2\otimes1+1\otimes(-\triangle_{\rm N}),
\end{equation}
with $P$ the self-adjoint operator of differentiation on $\T$ with periodic boundary
condition and $-\triangle_{\rm N}$ the Neumann Laplacian on $\R_+$. Since $(P+k)^2$
has purely discrete spectrum given by eigenvalues $\lambda_{k,n}:=(n+k)^2$, $n\in\Z$,
and since $-\triangle_{\rm N}$ has purely absolutely continuous spectrum
$\sigma(-\triangle_{\rm N})=[0,\infty)$, the operator $H^0_k$ has purely absolutely
continuous spectrum $\sigma(H^0_k)=[k^2,\infty)$ and its spectral multiplicity is
piecewise constant with a jump at each point of the threshold set
$$
\tau_k:=\big\{\lambda_{k,n}\big\}_{n\in\Z}\;\!.
$$
A set of normalized eigenvectors for the operator $(P+k)^2$ is given by the family
$\big\{\frac1{\sqrt{2\pi}}\e^{in(\,\cdot\,)}\big\}_{n\in\Z}\subset\ltwo(\T)$.
Since this family is independent of $k$, we simply write $\{\P_n\}_{n\in\Z}$ for the
corresponding set of one-dimensional orthogonal projections in $\ltwo(\T)$.

\subsection{Spectral representation for $\boldsymbol{H_k^0}$}

We now give a spectral representation of the operator $H^0_k$ defined in \eqref{defH0}
(see \cite[Sec.~2.2]{Fra06} for the original representation). For that purpose, we fix
$k\in[-1/2,1/2]$ and define the Hilbert spaces
$$
\Hrond_{k,n}:=\ltwo\big([\lambda_{k,n},\infty);\P_n\;\!\ltwo(\T)\big)
\quad\hbox{and}\quad
\Hrond_k:=\bigoplus_{n\in\Z}\Hrond_{k,n}.
$$
We set
$
\S(\R_+):=\big\{\eta\mid\eta=\zeta|_{\R_+}\hbox{ for some }\zeta\in\S(\R)\big\},
$
we let $\Fc:\ltwo(\R_+)\to\ltwo(\R_+)$ be the unitary cosine transform given by
\begin{equation}\label{eq_cosine}
(\Fc\;\!\eta)(y)
:=\left(\frac2{\pi}\right)^{1/2}\int_0^\infty\cos(yx)\;\!\eta(x)\,\d x,
\quad\eta\in\S(\R_+),~y\in\R_+,
\end{equation}
and we let $\U_k:\ltwo(\Pi)\to\Hrond_k$ be the unitary operator given for each
$\varphi\in\ltwo(\T)\odot\S(\R_+)$ by
$$
(\U_k\;\!\varphi)_n(\lambda)
:=2^{-1/2}(\lambda-\lambda_{k,n})^{-1/4}
\big((\P_n\otimes \Fc)\;\!\varphi\big)
\big(\;\cdot\;,\sqrt{\lambda-\lambda_{k,n}}\big),
\quad n\in\Z,~\lambda>\lambda_{k,n}.
$$
Then, the operator $\U_k$ is a spectral transformation for $H^0_k$ in the sense that
$\U_kH^0_k\;\!\U_k^*=L_k$, with $L_k$ the maximal multiplication operator in
$\Hrond_k$ given by
$$
(L_k\xi)_n(\lambda):=\lambda\;\!\xi_n(\lambda),
\quad\xi\in\dom(L_k)
:=\Bigg\{\xi\in\Hrond_k\mid\sum_{n\in\Z}\int_{\lambda_{k,n}}^\infty
\lambda^2\;\!\|\xi_n(\lambda)\|^2_{\ltwo(\T)}\,\d\lambda<\infty\Bigg\},
~n\in\Z,~\lambda >\lambda_{k,n}.
$$

The operator $\U_k$ satisfies the following regularity properties\hspace{1pt}: If we
define the weighted spaces
$$
\H_s(\R_+)
:=\big\{\eta\in\ltwo(\R_+)\mid\langle X\rangle^s\;\!\eta\in\ltwo(\R_+)\big\},
\quad s\ge0,
$$
with $X$ the maximal operator of multiplication by the variable in $\ltwo(\R_+)$ and
$\langle x\rangle:=(1+x^2)^{1/2}$, then the operator
$$
\U_k(n,\lambda)\;\!\varphi
:=(\U_k\;\!\varphi)_n(\lambda),
\quad n\in\Z,~\lambda>\lambda_{k,n},~\varphi\in\ltwo(\T)\odot\S(\R_+),
$$
extends to an element of $\B\big(\ltwo(\T)\otimes\H_s(\R_+);\P_n\;\!\ltwo(\T)\big)$
for each $s>1/2$, and the map
$$
(\lambda_{k,n},\infty)\ni\lambda\mapsto\U_k(n,\lambda)
\in\B\big(\ltwo(\T)\otimes\H_s(\R_+);\P_n\;\!\ltwo(\T)\big)
$$
is continuous (see for example \cite[Prop.~2.5]{Tie06} for an analogue of these
results on $\R$ instead of $\R_+$).

\section{Spectral analysis of $\boldsymbol{H^V_k}$}\label{sec_spectral}
\setcounter{equation}{0}

In this section, we give some information on the eigenvalues of $H^V_k$, and we derive
resolvent expansions at embedded thresholds and eigenvalues for $H^V_k$ for any fixed
value of $k\in[-1/2,1/2]$.

Following the standard idea of decomposing the perturbation into factors, we define
the functions
$$
v:\T\to\R,\quad\theta\mapsto|V(\theta)|^{1/2}
\qquad\hbox{and}\qquad
u:\T\to\{-1,1\},\quad\theta\mapsto
\begin{cases}
1  & \hbox{if}~~V(\theta)\ge0\\
-1 & \hbox{if}~~V(\theta)<0.
\end{cases}
$$
Also, we use the same notation for a function and for the corresponding operator
of multiplication, and we note that $u$ is both unitary and self-adjoint as a
multiplication operator in $\ltwo(\T)$. Moreover, we set $R^0_k(z):=(H^0_k-z)^{-1}$ and
$R^V_k(z)=(H^V_k-z)^{-1}$ for $z\in\C\setminus\R$, and we define the operator
$G\in\B\big(\widetilde\H^1(\Pi);\ltwo(\T)\big)$ by
$$
(G\varphi)(\theta):= v(\theta)\;\!\varphi(\theta,0),
\quad\theta\in\T.
$$
Then, the operator $u+GR^0_k(z)G^*$ has a
bounded inverse in $\ltwo(\T)$ for each $z\in\C\setminus\R$,
and the resolvent equation may be written as (see \cite[Prop.~3.1]{Fra03})
\begin{equation}\label{eq_resolv}
R^V_k(z)
=R^0_k(z)-R^0_k(z)G^*\big(u+GR^0_k(z)G^*\big)^{-1}GR^0_k(z),\quad z\in\C\setminus\R.
\end{equation}
Alternatively, one can deduce from \cite[Eq.~(1.9.14)]{Yaf92} the equivalent formula
\begin{equation}\label{eqlink}
G R^V_k(z) G^*
=u-u\;\!\big(u+GR^0_k(z)G^*\big)^{-1}u.
\end{equation}

In view of these equalities, our goal reduces to derive asymptotic expansions for the
operator $\big(u+GR^0_k(z)G^*\big)^{-1}$ as $z\to z_0\in\tau_k\cup\sigma_{\rm p}(H^V_k)$.
For this, we first choose the square root $\sqrt z$
of $z\in\C\setminus[0,\infty)$ such that $\im(\sqrt z)>0$, and then use this convention to
compute explicitly the kernel of the operator $R^0_k(z):$
$$
\big(R^0_k(z)\big)(\theta,x,\theta',x')
=\frac i{4\pi}\sum_{n\in\Z}\frac{\e^{in(\theta-\theta')}}{\sqrt{z-\lambda_{k,n}}}
\big(\e^{i\sqrt{z-\lambda_{k,n}}\;\!(x+x')}+\e^{i\sqrt{z-\lambda_{k,n}}\;\!|x-x'|}\big),
\quad(\theta,x),(\theta',x')\in\Pi,
$$
(see \cite[Eq.~(3.1)]{Fra03} for a similar formula).
A straightforward computation then leads to the equality
\begin{equation}\label{toutreduit}
GR^0_k(z)G^*=i\;\!\sum_{n\in\Z}\frac{v\;\!\P_nv}{\sqrt{z-\lambda_{k,n}}}\;\!,
\quad z\in\C\setminus\R.
\end{equation}

In the sequel, we also use for $\lambda\in[k^2,\infty)$ the definitions
$$
\Z_k(\lambda):=\big\{n\in\Z\mid\lambda_{k,n}\le\lambda\big\},\quad
\Z_k(\lambda)^\bot:=\Z\setminus\Z_k(\lambda)\quad\hbox{and}\quad
\beta_{k,n}(\lambda):=|\lambda-\lambda_{k,n}|^{1/4},
$$
whose interest come from the following equalities\hspace{1pt}:
\begin{equation}\label{broccoli}
GR^0_k(\lambda+i\;\!0)G^*
:=\ulim_{\varepsilon\searrow0}GR^0_k(\lambda+i\;\!\varepsilon)G^*
=\sum_{n\in\Z_k(\lambda)^\bot}\frac{v\;\!\P_n v}{\beta_{k,n}(\lambda)^2}
+i\sum_{n\in\Z_k(\lambda)}\frac{v\;\!\P_n v}{\beta_{k,n}(\lambda)^2}\;\!,
\quad\lambda\in\R\setminus\tau_k,
\end{equation}
where the convergence is uniform on compact subsets of $\R\setminus\tau_k$.

\begin{Lemma}\label{lem_vp}
Assume that $V\in\linf(\R;\R)$ is $2\pi$-periodic. Then, a value
$\lambda\in\R\setminus\tau_k$ is an eigenvalue of $H^V_k$ if and only if
$$
\K:=\ker\Bigg(u+\sum_{n\in\Z_k(\lambda)^\bot}\frac{v\;\!\P_n v}
{\beta_{k,n}(\lambda)^2}\Bigg)
\bigcap\big(\cap_{n\in\Z_k(\lambda)}\ker(\P_nv)\big)\ne\{0\},
$$
and in this case the multiplicity of $\lambda$ is equal to the dimension of $\K$.
\end{Lemma}

\begin{proof}
We apply \cite[Lemma~4.7.8]{Yaf92}. Once the assumptions of this lemma are checked, it
implies that the multiplicity of an eigenvalue
$\lambda\in\sigma_{\rm p}(H^V_k)\setminus\tau_k$ is equal to the multiplicity of the
eigenvalue $1$ of the operator $-GR^0_k(\lambda+i0)G^*u$. But, the unitarity and the
self-adjointness of $u$ together with the equality \eqref{broccoli} imply that the
following conditions are equivalent for $q\in\ltwo(\T)$\hspace{1pt}:
$$
-GR^0_k(\lambda+i\hspace{1pt}0)G^*uq=q
\iff uq\in\ker\Bigg(u+\sum_{n\in\Z_k(\lambda)^\bot}\frac{v\;\!\P_n v}
{\beta_{k,n}(\lambda)^2}+i\sum_{n\in\Z_k(\lambda)}\frac{v\;\!\P_n v}
{\beta_{k,n}(\lambda)^2}\Bigg),
$$
and the second condition is in turn equivalent to the inclusion $uq\in\mathcal K$.
Thus, since $u$ is unitary we are left in proving that the assumptions of
\cite[Lemma~4.7.8]{Yaf92} hold in a neighbourhood of
$\lambda\in\sigma_{\rm p}(H^V_k)\setminus\tau_k$.

Since the multiplicity of the spectrum of $H^0_k$ is constant in each small enough
neighbourhood of $\lambda\in\sigma_{\rm p}(H^V_k)\setminus\tau_k$, it is sufficient to
prove that the operators $G$ and $u\;\!G$ are strongly $H^0_k$-smooth with some
exponent $\alpha>1/2$ on any compact subinterval of $\R\setminus\tau_k$ (see
\cite[Def.~4.4.5]{Yaf92} for the definition of strong $H_0$-smoothness). However, such
a property can be checked either by using \cite[Lemma 2.3]{Fra06} or by using the
explicit formula
$$
(\U_kG^*q)_n(\lambda)
=\pi^{-1/2}\beta_{k,n}(\lambda)^{-1}\;\!\P_n(vq)\in\P_n\;\!\ltwo(\T),
\quad n\in\Z,~\lambda >\lambda_{k,n},~q\in\ltwo(\T),
$$
and the same formula with $G^*$ replaced by $G^*u$.
\end{proof}

Lemma \ref{lem_vp} has simple, but interesting, consequences on the localization of
the eigenvalues of $H^V_k$. Indeed, one has for each $\lambda\in\R\setminus\tau_k$ the
inequality
\begin{equation}\label{ineqrough}
\left\|\sum_{n\in\Z_k(\lambda)^\bot}\frac{v\;\!\P_n v}
{\beta_{k,n}(\lambda)^2}\right\|
\le\sup_{n\in\Z_k(\lambda)^\bot}\frac{\|V\|_\infty}{\beta_{k,n}(\lambda)^2}\;\!.
\end{equation}
Therefore, if $m\in\Z$ is such that $[\lambda,\lambda_{k,m})\cap\tau_k=\varnothing$
and $\lambda_{k,m}-\lambda>\|V\|_\infty^2$, one infers from \eqref{ineqrough} and
\cite[Thm.~IV.1.16]{Kato} that the subspace $\K\equiv\K(\lambda)$ of Lemma
\ref{lem_vp} is trivial. In other words, the possible eigenvalues of $H^V_k$ can only
be located at a finite distance (independent of $m$) on the left of each threshold. On
the other hand, since the distance between two consecutive thresholds $\lambda_{k,m}$
and $\lambda_{k,m'}$ is proportional to $|m|$, the interval free of possible
eigenvalues between two consecutive thresholds is increasing as $|m|\to\infty$.

\begin{Remark}
The above localization result is sharp. Indeed, if $V$ is a constant function with
$V<0$, then we know from \cite[Ex.~4.2]{FS04} that
$\sigma_{\rm p}(H^V_k)=\big\{\lambda_{k,m}-V^2\mid m\in\Z\big\}$.
\end{Remark}

\subsection{Resolvent expansions for $\boldsymbol{H^V_k}$}\label{secresolvex}

We are now ready to derive the resolvent expansions at all points of interest by using
the iterative procedure of \cite[Sec.~3.1]{RT14} and the associated inversion
formulas. For that purpose, we set $\C_+:=\{z\in\C\mid\im(z)>0\}$ and we adapt a
convention of \cite{JN01} by considering values $z=\lambda-\kappa^2$ with $\kappa$
belonging to the set
$$
O(\varepsilon)
:=\big\{\kappa\in\C\mid|\kappa|\in(0,\varepsilon),~\re(\kappa)>0\hbox{ and }
\im(\kappa)<0\big\},\quad\varepsilon>0,
$$
or the set
$$
\widetilde O(\varepsilon)
:=\big\{\kappa\in\C\mid|\kappa|\in(0,\varepsilon),~\re(\kappa)\ge0\hbox{ and }
\im(\kappa)\le0\big\},\quad\varepsilon>0.
$$
Note that if $\kappa \in O(\varepsilon)$, then $-\kappa^2\in \C_+$ while if
$\kappa\in\widetilde O(\varepsilon)$, then $-\kappa^2\in \overline{\C_+}$. With these
notations at hand, the main result of this section reads as follows\hspace{1pt}:

\begin{Proposition}\label{Prop_Asymp}
Suppose that $V\in\linf(\R;\R)$ is $2\pi$-periodic, fix
$\lambda\in\tau_k\cup\sigma_{\rm p}(H^V_k)$, and take $\kappa\in O(\varepsilon)$ with
$\varepsilon>0$ small enough. Then, the operator
$\big(u+GR^0_k(\lambda-\kappa^2)G^*\big)^{-1}$ belongs to $\B\big(\ltwo(\T)\big)$ and
is continuous in the variable $\kappa\in O(\varepsilon)$. Moreover, the continuous
function
$$
O(\varepsilon)\ni\kappa\mapsto
\big(u+GR^0_k(\lambda-\kappa^2)G^*\big)^{-1}\in\B\big(\ltwo(\T)\big)
$$
extends continuously to a function
$
\widetilde O(\varepsilon)\ni\kappa\mapsto
\M_k(\lambda,\kappa)\in\B\big(\ltwo(\T)\big)
$,
and for each $\kappa\in\widetilde O(\varepsilon)$ the operator $\M_k(\lambda,\kappa)$
admits an asymptotic expansion in $\kappa$. The precise form of this expansion is
given on the r.h.s. of the equations \eqref{sol1} and \eqref{eq_expansion_2} below.
\end{Proposition}

We recall that the relation between the asymptotic expansions given of Proposition
\ref{Prop_Asymp} and the resolvent of $H^V_k$ is given by formula \eqref{eqlink}. The
proof of Proposition \ref{Prop_Asymp} is mainly based on an inversion formula which we
reproduce here for completeness (see also \cite[Prop.~1]{JN04} for an earlier
version)\hspace{1pt}:

\begin{Proposition}[Prop. 2.1 of \cite{RT14}]\label{propourversion}
Let $O\subset\C$ be a subset with $0$ as an accumulation point, and let $\H$ be an
Hilbert space. For each $z\in O$, let $A(z)\in\B(\H)$ satisfy
$$
A(z)=A_0+z\hspace{1pt}A_1(z),
$$
with $A_0\in\B(\H)$ and $\|A_1(z)\|_{\B(\H)}$ uniformly bounded as $z\to0$. Let also
$S\in\B(\H)$ be a projection such that (i) $A_0 + S$ is invertible with bounded
inverse, (ii) $S(A_0+S)^{-1}S=S$. Then, for $|z|>0$ small enough the operator
$B(z):S\H\to S\H$ defined by
$$
B(z)
:=\frac1z\left(S-S\big(A(z)+S\big)^{-1}S\right)
\equiv S(A_0+S)^{-1}\Bigg(\sum_{j\ge0}(-z)^j\big(A_1(z)(A_0+S)^{-1}\big)^{j+1}\Bigg)S
$$
is uniformly bounded as $z\to0$. Also, $A(z)$ is invertible in $\H$ with bounded
inverse if and only if $B(z)$ is invertible in $S\H$ with bounded inverse, and in this
case one has
$$
A(z)^{-1}
=\big(A(z)+S\big)^{-1}+\frac1z\big(A(z)+S\big)^{-1}SB(z)^{-1}S\big(A(z)+S\big)^{-1}.
$$
\end{Proposition}

\begin{proof}[Proof of Proposition \ref{Prop_Asymp}]
For each $\lambda\in\R$, $\varepsilon>0$ and $\kappa\in O(\varepsilon)$, one has
$\im(\lambda-\kappa^2)\ne0$. Thus, the operator
$\big(u+GR^0_k(\lambda-\kappa^2)G^*\big)^{-1}$ belongs to $\B\big(\ltwo(\T)\big)$ and
is continuous in $\kappa\in O(\varepsilon)$ due to \eqref{eq_resolv}. For the other
claims, we distinguish the cases $\lambda\in\tau_k$ and
$\lambda\in\sigma_{\rm p}(H^V_k)\setminus\tau_k$, treating first the case
$\lambda\in\tau_k$. All the operators defined below depend on the choice of $\lambda$,
but for simplicity we do not always mention these dependencies.

(i) Assume that $\lambda\in\tau_k$, take $\varepsilon>0$, set
$N:=\{n\in\Z\mid\lambda_{k,n}=\lambda\}$, and write $\P:=\sum_{n \in N}\P_n$ for the
(one or two-dimensional) orthogonal projection associated with the eigenvalue
$\lambda$ of the operator $(P+k)^2$. Then, \eqref{toutreduit} implies for
$\kappa\in O(\varepsilon)$ that
$$
\big(u+GR^0_k(\lambda-\kappa^2)G^*\big)^{-1}
=\kappa\left\{v\;\!\P v
+\kappa\left(u+i\;\!\sum_{n\notin N}\frac{v\;\!\P_n v}
{\sqrt{\lambda-\kappa^2-\lambda_{k,n}}}\right)\right\}^{-1}.
$$
Moreover, direct computations show that the function
$$
O(\varepsilon)\ni\kappa\mapsto
u+i\;\!\sum_{n\notin N}\frac{v\;\!\P_n v}{\sqrt{\lambda-\kappa^2-\lambda_{k,n}}}
\in\B\big(\ltwo(\T)\big)
$$
extends continuously to a function
$
\widetilde O(\varepsilon)\ni\kappa\mapsto
M_1(\kappa)\in\B\big(\ltwo(\T)\big)
$
with $\|M_1(\kappa)\|_{\B(\ltwo(\T))}$ uniformly bounded as $\kappa\to0$. Thus, one
has for each $\kappa\in O(\varepsilon)$
\begin{equation}\label{form_I_0}
\big(u+GR^0_k(\lambda-\kappa^2)G^*\big)^{-1}=\kappa\;\!I_0(\kappa)^{-1}
\quad\hbox{with}\quad
I_0(\kappa):=v\;\!\P v+\kappa\;\!M_1(\kappa).
\end{equation}
Now, since $N_0:=I_0(0)=v\;\!\P v$ is a finite-rank operator, $0$ is not a limit point
of its spectrum. Also, $N_0$ is self-adjoint, therefore the orthogonal projection
$S_0$ on $\ker(N_0)$ is equal to the Riesz projection of $N_0$ associated with the
value $0$. We can thus apply Proposition \ref{propourversion} (see
\cite[Cor.~2.8]{RT14}), and obtain for $\kappa\in\widetilde O(\varepsilon)$ with
$\varepsilon>0$ small enough that the operator
$I_1(\kappa):S_0\ltwo(\T)\to S_0\ltwo(\T)$ defined by
\begin{equation}\label{defI1}
I_1(\kappa):=\sum_{j\ge0}(-\kappa)^jS_0\;\!
\big\{M_1(\kappa)\big(I_0(0)+S_0\big)^{-1}\big\}^{j+1}S_0
\end{equation}
is uniformly bounded as $\kappa\to0$. Furthermore, $I_1(\kappa)$ is invertible in
$S_0\ltwo(\T)$ with bounded inverse satisfying the equation
$$
I_0(\kappa)^{-1}
=\big(I_0(\kappa)+S_0\big)^{-1}+\frac 1{\kappa}\;\!\big(I_0(\kappa)+S_0\big)^{-1}
S_0I_1(\kappa)^{-1}S_0\big(I_0(\kappa)+S_0\big)^{-1}.
$$
It follows that for $\kappa\in O(\varepsilon)$ with $\varepsilon>0$ small enough, one
has
\begin{equation}\label{eq18}
\big(u+GR^0_k(\lambda-\kappa^2)G^*\big)^{-1}
=\kappa\;\!\big(I_0(\kappa)+S_0\big)^{-1}+
\big(I_0(\kappa)+S_0\big)^{-1}S_0I_1(\kappa)^{-1}S_0\big(I_0(\kappa)+S_0\big)^{-1},
\end{equation}
with the first term vanishing as $\kappa \to0$.

To describe the second term of $\big(u+GR^0_k(\lambda-\kappa^2)G^*\big)^{-1}$ as
$\kappa\to0$ we note that the equality $\big(I_0(0)+S_0\big)^{-1}S_0=S_0$ and the
definition \eqref{defI1} imply for $\kappa\in\widetilde O(\varepsilon)$ with
$\varepsilon>0$ small enough that
\begin{equation}\label{form_I_1}
I_1(\kappa)=S_0M_1(0)S_0+\kappa\;\!M_2(\kappa),
\end{equation}
with
\begin{align*}
M_2(\kappa)
&:=\frac i\kappa\;\!S_0\sum_{n\notin N}
\left(\frac1{\sqrt{\lambda-\kappa^2-\lambda_{k,n}}}
-\frac1{\sqrt{\lambda-\lambda_{k,n}}}\right)v\;\!\P_n vS_0\\
&\quad-\sum_{j\ge0}(-\kappa)^jS_0\;\!
\big\{M_1(\kappa)\big(I_0(0)+S_0\big)^{-1}\big\}^{j+2}S_0.
\end{align*}
Also, we note that the expansion
\begin{equation}\label{eq23}
\frac1{\sqrt{\lambda-\kappa^2-\lambda_{k,n}}}
=\frac1{\sqrt{\lambda-\lambda_{k,n}}}
\left(1+\frac{\kappa^2}{2(\lambda-\lambda_{k,n})}+\O(\kappa^4)\right),
\quad n\notin N,
\end{equation}
implies that $\|M_2(\kappa)\|_{\B(S_0\ltwo(\T))}$ is uniformly bounded as
$\kappa\to0$.

Now, we have
\begin{equation}\label{eqM10}
M_1(0)
=u+\sum_{n\in\Z_k(\lambda)^\bot}\frac{v\;\!\P_n v}{\beta_{k,n}(\lambda)^2}
+i\sum_{n\in\Z_k(\lambda)^-}\frac{v\;\!\P_n v}{\beta_{k,n}(\lambda)^2}\;\!,
\end{equation}
with $\Z_k(\lambda)^-:=\{n\in\Z\mid\lambda_{k,n}<\lambda\}$. Therefore, $M_1(0)$ is
the sum of the unitary and self-adjoint operator $u$, the self-adjoint and compact
operator $\sum_{n\in\Z_k(\lambda)^\bot}\frac{v\;\!\P_n v}{\beta_{k,n}(\lambda)^2}$,
and the compact operator with non-negative imaginary part
$i\sum_{n\in\Z_k(\lambda)^-}\frac{v\;\!\P_n v}{\beta_{k,n}(\lambda)^2}$. So, since
$S_0$ is an orthogonal projection with finite-dimensional kernel, the operator
$I_1(0)=S_0M_1(0)S_0$ acting in the Hilbert space $S_0\ltwo(\T)$ can also be written
as the sum of a unitary and self-adjoint operator, a self-adjoint and compact
operator, and a compact operator with non-negative imaginary part. Thus, the result
\cite[Cor.~2.8]{RT14} applies with $S_1$ the finite-rank orthogonal projection on
$\ker\big(I_1(0)\big)$, and Proposition \ref{propourversion} can be applied to
$I_1(\kappa)$ as it was done for $I_0(\kappa)$.

Therefore, for $\kappa\in\widetilde O(\varepsilon)$ with $\varepsilon>0$ small enough,
the operator $I_2(\kappa):S_1\ltwo(\T)\to S_1\ltwo(\T)$ defined by
$$
I_2(\kappa)
:=\sum_{j\ge0}(-\kappa)^jS_1\big\{M_2(\kappa)\big(I_1(0)+S_1\big)^{-1}\big\}^{j+1}S_1
$$
is uniformly bounded as $\kappa\to0$. Furthermore, $I_2(\kappa)$ is invertible in
$S_1\ltwo(\T)$ with bounded inverse satisfying the equation
$$
I_1(\kappa)^{-1}
=\big(I_1(\kappa)+S_1\big)^{-1}+\frac1\kappa\;\!\big(I_1(\kappa)+S_1\big)^{-1}
S_1I_2(\kappa)^{-1}S_1\big(I_1(\kappa)+S_1\big)^{-1}.
$$
This expression for $I_1(\kappa)^{-1}$ can now be inserted in \eqref{eq18} in order to
get for $\kappa\in O(\varepsilon)$ with $\varepsilon>0$ small enough
\begin{align}
&\big(u+GR^0_k(\lambda-\kappa^2)G^*\big)^{-1}\nonumber\\
&=\kappa\;\!\big(I_0(\kappa)+S_0\big)^{-1}
+\big(I_0(\kappa)+S_0\big)^{-1}S_0\big(I_1(\kappa)+S_1\big)^{-1}S_0
\big(I_0(\kappa)+S_0\big)^{-1}\nonumber\\
&\quad+\frac1\kappa\;\!\big(I_0(\kappa)+S_0\big)^{-1}
S_0\big(I_1(\kappa)+S_1\big)^{-1}S_1I_2(\kappa)^{-1}S_1\big(I_1(\kappa)+S_1\big)^{-1}
S_0\big(I_0(\kappa)+S_0\big)^{-1},\label{eq_F_second}
\end{align}
with the first two terms bounded as $\kappa\to0$.

We now concentrate on the last term and check once more that the assumptions of
Proposition \ref{propourversion} are satisfied. For this, we recall that
$\big(I_1(0)+S_1\big)^{-1}S_1=S_1$, and observe that for
$\kappa\in\widetilde O(\varepsilon)$ with $\varepsilon>0$ small enough
\begin{equation}\label{form_I_2}
I_2(\kappa)=S_1M_2(0)S_1+\kappa\;\!M_3(\kappa),
\end{equation}
with
$$
M_2(0)=-S_0M_1(0)\big(I_0(0)+S_0\big)^{-1}M_1(0)S_0
\quad\hbox{and}\quad
M_3(\kappa)\in\O(1).
$$
The inclusion $M_3(\kappa)\in\O(1)$ follows from simple computations taking the
expansion \eqref{eq23} into account. As observed above, one has $M_1(0)=Y+iZ^*Z$, with
$Y,Z$ bounded self-adjoint operators in $\ltwo(\T)$. Therefore,
$I_1(0)=S_0M_1(0)S_0=S_0YS_0+i(ZS_0)^*(Z S_0)$, and one infers from
\cite[Cor.~2.5]{RT14} that $Z S_0S_1=0=S_1S_0 Z^*$. Since $S_1S_0=S_1=S_0S_1$, it
follows that $ZS_1=0=S_1Z^*$. Therefore, we have
\begin{align*}
I_2(0)
&=-S_1M_1(0)\big(I_0(0)+S_0\big)^{-1}M_1(0)S_1\\
&=-S_1(Y+iZ^*Z)\big(I_0(0)+S_0\big)^{-1}(Y+iZ^*Z)S_1\\
&=-S_1Y\big(I_0(0)+S_0\big)^{-1}YS_1,
\end{align*}
and thus $-I_2(0)$ is a positive operator.

Since $S_1\ltwo(\T)$ is finite-dimensional, $0$ is not a limit point of
$\sigma\big(I_2(0)\big)$. So, the orthogonal projection $S_2$ on
$\ker\big(I_2(0)\big)$ is a finite-rank operator, and Proposition \ref{propourversion}
applies to $I_2(0)+\kappa\hspace{1pt}M_3(\kappa)$. Thus, for
$\kappa\in\widetilde O(\varepsilon)$ with $\varepsilon>0$ small enough, the operator
$I_3(\kappa):S_2\ltwo(\T)\to S_2\ltwo(\T)$ defined by
$$
I_3(\kappa):=\sum_{j\ge0}(-\kappa)^jS_2
\;\!\big\{M_3(\kappa)\big(I_2(0)+S_2\big)^{-1}\big\}^{j+1}S_2
$$
is uniformly bounded as $\kappa\to0$. Furthermore, $I_3(\kappa)$ is invertible in
$S_2\ltwo(\T)$ with bounded inverse satisfying the equation
$$
I_2(\kappa)^{-1}
=\big(I_2(\kappa)+S_2\big)^{-1}
+\frac1\kappa\big(I_2(\kappa)+S_2\big)^{-1}S_2I_3(\kappa)^{-1}S_2
\big(I_2(\kappa)+S_2\big)^{-1}.
$$
This expression for $I_2(\kappa)^{-1}$ can now be inserted in \eqref{eq_F_second} in
order to get for $\kappa\in O(\varepsilon)$ with $\varepsilon>0$ small enough
\begin{align}
&\big(u+GR^0_k(\lambda-\kappa^2)G^*\big)^{-1}\nonumber\\
&=\kappa\big(I_0(\kappa)+S_0\big)^{-1}
+\big(I_0(\kappa)+S_0\big)^{-1}S_0\big(I_1(\kappa)+S_1\big)^{-1}S_0
\big(I_0(\kappa)+S_0\big)^{-1}\nonumber\\
&\quad+\frac1\kappa\big(I_0(\kappa)+S_0\big)^{-1}S_0
\big(I_1(\kappa)+S_1\big)^{-1}S_1\big(I_2(\kappa)+S_2\big)^{-1}S_1
\big(I_1(\kappa)+S_1\big)^{-1}S_0\big(I_0(\kappa)+S_0\big)^{-1}\nonumber\\
&\quad+\frac1{\kappa^2}\big(I_0(\kappa)+S_0\big)^{-1}S_0
\big(I_1(\kappa)+S_1\big)^{-1}S_1\big(I_2(\kappa)+ S_2\big)^{-1}S_2 I_3(\kappa)^{-1}
S_2\big(I_2(\kappa)+S_2\big)^{-1}S_1\nonumber\\
&\qquad\times\big(I_1(\kappa)+S_1\big)^{-1}S_0\big(I_0(\kappa)+S_0\big)^{-1}.
\label{sol1}
\end{align}

Fortunately, the iterative procedure stops here. The argument is based on the relation \eqref{eqlink}
and the fact that $H^V_k$ is a self-adjoint operator. Indeed, if we choose
$\kappa=\frac\varepsilon2(1-i)\in O(\varepsilon)$, then the inequality
$\big\|\kappa^2(H^V_k-\lambda+\kappa^2)^{-1}\big\|_{\B(\ltwo(\Pi))}\le1$ holds, and
thus
\begin{equation}\label{notresauveur}
\limsup_{\kappa\to0}
\big\|\kappa^2\big(u+GR^0_k(\lambda-\kappa^2)G^*\big)^{-1}\big\|_{\B(\ltwo(\T))}
<\infty.
\end{equation}
So, if we replace $\big(u+GR^0_k(\lambda-\kappa^2)G^*\big)^{-1}$ by the expression
\eqref{sol1} and if we take into account that all factors of the form
$\big(I_j(\kappa)+S_j\big)^{-1}$ have a finite limit as $\kappa\to0$, we infer from
\eqref{notresauveur} that
\begin{equation}\label{notresecondsauveur}
\limsup_{\kappa\to0}\big\|I_3(\kappa)^{-1}\big\|_{\B(S_2\ltwo(\T))}<\infty.
\end{equation}
Therefore, it only remains to show that this relation holds not just for
$\kappa=\frac\varepsilon2(1-i)$ but for all $\kappa\in\widetilde O(\varepsilon)$. For
that purpose, we consider $I_3(\kappa)$ once again, and note that
\begin{equation}\label{souffrance1}
I_3(\kappa)=S_2M_3(0)S_2+\kappa\;\!M_4(\kappa)
\quad\hbox{with}\quad M_4(\kappa)\in\O(1).
\end{equation}
The precise form of $M_3(0)$ can be computed explicitly, but is irrelevant.

Now, since $I_3(0)$ acts in a finite-dimensional space, $0$ is an isolated eigenvalue
of $I_3(0)$ if $0\in\sigma\big(I_3(0)\big)$, in which case we write $S_3$ for the
corresponding Riesz projection. Then, the operator $I_3(0)+S_3$ is invertible with
bounded inverse, and \eqref{souffrance1} implies that $I_3(\kappa)+S_3$ is also
invertible with bounded inverse for $\kappa\in\widetilde O(\varepsilon)$ with
$\varepsilon>0$ small enough. In addition, one has
$\big(I_3(\kappa)+S_3\big)^{-1}=\big(I_3(0)+S_3\big)^{-1}+\O(\kappa)$. By the
inversion formula given in \cite[Lemma 2.1]{JN01}, one infers that
$S_3-S_3\big(I_3(\kappa)+S_3\big)^{-1}S_3$ is invertible in $S_3\ltwo(\T)$ with
bounded inverse and that the following equalities hold
\begin{align*}
I_3(\kappa)^{-1}
&=\big(I_3(\kappa)+S_3\big)^{-1}+\big(I_3(\kappa)+S_3\big)^{-1}S_3
\big\{S_3-S_3\big(I_3(\kappa)+S_3\big)^{-1}S_3\big\}^{-1}S_3
\big(I_3(\kappa)+S_3\big)^{-1}\\
&=\big(I_3(\kappa)+S_3\big)^{-1}+\big(I_3(\kappa)+S_3\big)^{-1}S_3
\big\{S_3-S_3\big(I_3(0)+S_3\big)^{-1}S_3+\O(\kappa)\big\}^{-1}S_3
\big(I_3(\kappa)+S_3\big)^{-1}.
\end{align*}
This implies that \eqref{notresecondsauveur} holds for some
$\kappa\in\widetilde O(\varepsilon)$ if and only if the operator
$S_3-S_3\big(I_3(0)+S_3\big)^{-1}S_3$ is invertible in $S_3\ltwo(\T)$ with bounded
inverse. But, we already know from what precedes that \eqref{notresecondsauveur} holds
for $\kappa=\frac\varepsilon2(1-i)$. So, the operator
$S_3-S_3\big(I_3(0)+S_3\big)^{-1}S_3$ is invertible in $S_3\ltwo(\T)$ with bounded
inverse, and thus \eqref{notresecondsauveur} holds for all
$\kappa\in\widetilde O(\varepsilon)$. Therefore, \eqref{sol1} implies that the
function
$$
O(\varepsilon)\ni\kappa\mapsto
\big(u+GR^0_k(\lambda-\kappa^2)G^*\big)^{-1}\in\B\big(\ltwo(\T)\big)
$$
extends continuously to the function
$
\widetilde O(\varepsilon)\ni\kappa\mapsto\M_k(\lambda,\kappa)\in\B\big(\ltwo(\T)\big)
$,
with $\M_k(\lambda,\kappa)$ given by the r.h.s. of \eqref{sol1}.

(ii) Assume that $\lambda\in\sigma_{\rm p}(H^V_k)\setminus\tau_k$, take
$\varepsilon>0$, let $\kappa\in\widetilde O(\varepsilon)$, and set
$J_0(\kappa):=T_0+\kappa^2\hspace{1pt}T_1(\kappa)$ with
$$
T_0:=u+\sum_{n\in\Z_k(\lambda)^\bot}\frac{v\;\!\P_n v}{\beta_{k,n}(\lambda)^2}
+i\sum_{n\in\Z_k(\lambda)}\frac{v\;\!\P_n v}{\beta_{k,n}(\lambda)^2}
$$
and
$$
T_1(\kappa)
:=\frac i{\kappa^2}\sum_{n\in \Z}
\left(\frac1{\sqrt{\lambda-\kappa^2-\lambda_{k,n}}}
-\frac1{\sqrt{\lambda-\lambda_{k,n}}}\right)v\;\!\P_n v.
$$
Then, one infers from the expansion \eqref{eq23} that
$\|T_1(\kappa)\|_{\B(\ltwo(\T))}$ is uniformly bounded as $\kappa\to0$. Also, the
assumptions of \cite[Cor.~2.8]{RT14} hold for the operator $T_0$, and thus the Riesz
projection $S$ associated with the value $0\in\sigma(T_0)$ is an orthogonal
projection. It thus follows from Proposition \ref{propourversion} that for
$\kappa\in\widetilde O(\varepsilon)$ with $\varepsilon>0$ small enough, the operator
$J_1(\kappa):S\;\!\ltwo(\T)\to S\;\!\ltwo(\T)$ defined by
$$
J_1(\kappa)
:=\sum_{j\ge0}(-\kappa^2)^jS\;\!\big\{T_1(\kappa)(T_0+S)^{-1}\big\}^{j+1}S
$$
is uniformly bounded as $\kappa\to0$. Furthermore, $J_1(\kappa)$ is invertible in
$S\ltwo(\T)$ with bounded inverse satisfying the equation
$$
J_0(\kappa)^{-1}
=\big(J_0(\kappa)+S\big)^{-1}
+\frac1{\kappa^2}\big(J_0(\kappa)+S)^{-1}SJ_1(\kappa)^{-1}S
\big(J_0(\kappa)+S\big)^{-1}.
$$
It follows that for $\kappa\in O(\varepsilon)$ with $\varepsilon>0$ small enough one
has
\begin{equation}\label{sol2}
\big(u+GR^0_k(\lambda-\kappa^2)G^*\big)^{-1}
=\big(J_0(\kappa)+S\big)^{-1}
+\frac1{\kappa^2}\big(J_0(\kappa)+S)^{-1}SJ_1(\kappa)^{-1}S
\big(J_0(\kappa)+S\big)^{-1}.
\end{equation}
Fortunately, the iterative procedure already stops here. Indeed, the argument is
similar to the one presented above once we observe that
$$
J_1(\kappa)=ST_1(0)S+\kappa\hspace{1pt}T_2(\kappa)
\quad\hbox{with}\quad T_2(\kappa)\in\O(1).
$$
Therefore, \eqref{sol2} implies that the function
$$
O(\varepsilon)\ni\kappa\mapsto
\big(u+GR^0_k(\lambda-\kappa^2)G^*\big)^{-1}\in\B\big(\ltwo(\T)\big)
$$
extends continuously to the function
$
\widetilde O(\varepsilon)\ni\kappa\mapsto
\M_k(\lambda,\kappa)\in\B\big(\ltwo(\T)\big)
$,
with $\M_k(\lambda,\kappa)$ given by
\begin{equation}\label{eq_expansion_2}
\M_k(\lambda,\kappa)
=\big(J_0(\kappa)+S\big)^{-1}
+\frac1{\kappa^2}\big(J_0(\kappa)+S)^{-1}SJ_1(\kappa)^{-1}S
\big(J_0(\kappa)+S\big)^{-1}.
\end{equation}
\end{proof}

The non accumulation of eigenvalues of $H^V_k$ (except possibly at $+\infty$) can
easily be inferred from these asymptotic expansions (see for example
\cite[Corol.~3.3]{RT14} in the framework of quantum waveguides). However, since such a
result is already known in the present context \cite[Thm.~4.1]{FS04}, we do not prove
it again here.

We close this section with some auxiliary results which can all be deduced from the
expansions of Proposition \ref{Prop_Asymp}. The notations are borrowed from the proof
of Proposition \ref{Prop_Asymp} (with the only change that we extend by $0$ the
operators defined originally on subspaces of $\ltwo(\T)$ to get operators defined on
all of $\ltwo(\T)$). The proofs are skipped since they can be copied {\it mutatis
mutandis} from the corresponding ones in \cite[Sec.~3.1]{RT14}.

\begin{Lemma}\label{lemme_com}
Take $2\ge\ell\ge m\ge0$ and $\kappa\in\widetilde O(\varepsilon)$ with $\varepsilon>0$
small enough. Then, one has in $\B\big(\ltwo(\T)\big)$
$$
\big[S_\ell,\big(I_m(\kappa)+S_m\big)^{-1}\big]\in\O(\kappa).
$$
\end{Lemma}

Given $\lambda\in\tau_k$, we recall that
$N=\big\{n\in\Z\mid\lambda_{k,n}=\lambda\big\}$ and $\P=\sum_{n\in N}\P_n$.

\begin{Lemma}\label{relations_simples}
Take $\lambda\in\tau_k$ and let $Y$ be the real part of the operator $M_1(0)$.
\begin{enumerate}
\item[(a)] For each $n\in N$, one has $\P_n vS_0=0=S_0v\;\!\P_n$.
\item[(b)] For each $n\in \Z_k(\lambda)$, one has $\P_n vS_1=0=S_1v\;\!\P_n$.
\item[(c)] One has $YS_2=0=S_2Y$.
\item[(d)] One has $M_1(0)S_2=0=S_2M_1(0)$.
\end{enumerate}
\end{Lemma}

\section{Continuity properties of the scattering matrix}\label{seccont}
\setcounter{equation}{0}

We prove in this section continuity properties of the channel scattering matrices
associated with the scattering pair $\{H^0_k,H^V_k\}$. As before, the value of
$k\in[-1/2,1/2]$ is fixed throughout the section.

First, we note that under the assumption that $V\in\linf(\R;\R)$ is $2\pi$-periodic
the wave operators
$$
W_{k,\pm}:=\slim_{t\to\pm\infty}\e^{itH^V_k}\e^{-itH^0_k}
$$
exist and are complete (see \cite[Thm.~2.1]{Fra03}). As a consequence, the scattering
operator
$$
S_k:=W_{k,+}^*W_{k,-}
$$
is a unitary operator in $\ltwo(\Pi)$ which commutes with $H^0_k$, and thus $S_k$ is decomposable in the spectral representation of $H^0_k$.
To give an explicit formula for $S_k$ in that representation, that is, for the operator
$\U_kS_k\U_k^*$ in $\Hrond_k$, we recall from Proposition \ref{Prop_Asymp}, Lemma \ref{lem_vp},
and formula \eqref{broccoli}, that the operator
\begin{equation}\label{def_M0}
\M_k(\lambda,0)
\equiv\lim_{\varepsilon\searrow0}\big(u+GR^0_k(\lambda+i\varepsilon)G^*\big)^{-1}
\end{equation}
belongs to $\B\big(\ltwo(\T)\big)$ for each
$\lambda\in[k^2,\infty)\setminus\{\tau_k\cup\sigma_{\rm p}(H_k^V)\}$.
We also define for $n,n'\in\Z$ the operator
$\delta_{nn'}\in\B\big(\P_{n'}\;\!\ltwo(\T);\P_n\;\!\ltwo(\T)\big)$ by
$\delta_{nn'}=1$ if $n=n'$ and $\delta_{nn'}=0$ otherwise. Then, a computation
using stationary formulas as presented in \cite[Sec.~2.8]{Yaf92} shows that
$$
\big(\U_kS_k\U_k^*\xi\big)_n(\lambda)
:=\sum_{n'\in\Z_k(\lambda)}S_k(\lambda)_{nn'}\;\!\xi_{n'}(\lambda),
\quad \xi\in\Hrond_k,~n\in\Z,~\lambda\in[\lambda_{k,n},\infty)
\setminus\{\tau_k\cup\sigma_{\rm p}(H_k^V)\},
$$
with $S_k(\lambda)_{nn'}$ the channel scattering matrix given by
\begin{equation}\label{formule_S_matrix}
S_k(\lambda)_{nn'}
=\delta_{nn'}-2i\hspace{1pt}\beta_{k,n}(\lambda)^{-1}\;\!\P_nv\;\!\M_k(\lambda,0)
v\;\!\P_{n'}\beta_{k,n'}(\lambda)^{-1}.
\end{equation}
Moreover, the explicit formula \eqref{broccoli} implies for each $n,n'\in\Z$ the
continuity of the map
$$
[k^2,\infty)\setminus\{\tau_k\cup\sigma_{\rm p}(H_k^V)\}
\ni\lambda\mapsto S_k(\lambda)_{nn'}\in\B\big(\P_{n'}\ltwo(\T),\P_n\ltwo(\T)\big),
\quad\lambda_{k,n},\lambda_{k,n'}<\lambda.
$$
Therefore, in order to completely determine the continuity properties of the channel
scattering matrices $S_k(\lambda)_{nn'}$, it only remains to describe the behaviour of
$S_k(\lambda)_{nn'}$ as $\lambda\to\lambda_0\in\tau_k\cup\sigma_{\rm p}(H_k^V)$. In
the sequel, we consider separately the behaviour of $S_k(\lambda)_{nn'}$ at thresholds
and at embedded eigenvalues, starting with the thresholds.

For that purpose, we first note that for each $\lambda\in\tau_k$, a channel can either
be already open (in which case one has to show the existence and the equality of the
limits from the right and from the left), or can open at the energy $\lambda$ (in
which case one has only to show the existence of the limit from the right). Therefore,
as in the previous section, we shall fix $\lambda\in\tau_k$, and consider the
expression $S_k(\lambda-\kappa^2)_{nn'}$ for suitable $\kappa$ with $|\kappa|>0$ small
enough (recall that all expressions of Section \ref{sec_spectral} were also computed at fixed
$\lambda\in \tau_k$ but that the dependence on $\lambda$ has not been explicitly
written for the simplicity).

Before considering the continuity at thresholds, we define for each fixed
$\lambda\in\tau_k$, for $\kappa\in\widetilde O(\varepsilon)$ with $\varepsilon>0$
small enough, and for $2\ge\ell\ge m\ge0$ the operators
$$
C_{\ell m}(\kappa):=\big[S_\ell,\big(I_m(\kappa)+S_m\big)^{-1}\big]
\in\B\big(\ltwo(\T)\big),
$$
and note that $C_{\ell m}(\kappa)\in\O(\kappa)$ due to Lemma \ref{lemme_com}. In fact,
the formulas \eqref{form_I_0}, \eqref{form_I_1} and \eqref{form_I_2} imply that
$C_{\ell m}'(0):=\lim_{\kappa\to0}\frac1\kappa\;\!C_{\ell m}(\kappa)$ exists in
$\B\big(\ltwo(\T)\big)$. In other cases, we use the notation
$F(\kappa)\in\Oa(\kappa^n)$, $n\in\N$, for an operator $F(\kappa)\in\O(\kappa^n)$ such
that $\lim_{\kappa\to0}\kappa^{-n}F(\kappa)$ exists in $\B\big(\ltwo(\T)\big)$. We
also note that if $\kappa\in(0,\varepsilon)$ or $i\kappa\in(0,\varepsilon)$
with $\varepsilon>0$, then
$\kappa\in\widetilde O(\varepsilon)$ and
$-\kappa^2\in(-\varepsilon^2\varepsilon^2)\setminus\{0\}$.

\begin{Proposition}\label{propcont}
Assume that $V\in\linf(\R;\R)$ is $2\pi$-periodic, let $\lambda\in\tau_k$, take
$\kappa\in(0,\varepsilon)$ or $i\kappa\in(0,\varepsilon)$ with $\varepsilon>0$ small
enough, and let $n,n'\in\Z$.
\begin{enumerate}
\item[(a)] If $\lambda_{k,n},\lambda_{k,n'}<\lambda$, then the limit
$\lim_{\kappa\to0}S_k(\lambda-\kappa^2)_{nn'}$ exists and is given by
$$
\lim_{\kappa\to0}S_k(\lambda-\kappa^2)_{nn'}
=\delta_{nn'}-2i\beta_{k,n}(\lambda)^{-1}\;\!\P_nvS_0\big(I_1(0)+S_1\big)^{-1}S_0
v\;\!\P_{n'}\beta_{k,n'}(\lambda)^{-1}.
$$
\item[(b)] If $\lambda_{k,n},\lambda_{k,n'}\le\lambda$ and $-\kappa^2>0$, then the
limit $\lim_{\kappa\to0}S_k(\lambda-\kappa^2)_{nn'}$ exists and is given by
$$
\lim_{\kappa\to0}S_k(\lambda-\kappa^2)_{nn'}=
\begin{cases}
0 & \hbox{if~~$\lambda_{k,n}<\lambda$, $\lambda_{k,n'}=\lambda$,}\\
0 & \hbox{if~~$\lambda_{k,n}=\lambda$, $\lambda_{k,n'}<\lambda$,}\\
\delta_{nn'}-2\;\!\P_nv\big(I_0(0)+S_0\big)^{-1}v\;\!\P_{n'}\\
\quad+2\;\!\P_nv\;\!C_{10}'(0)S_1\big(I_2(0)+S_2\big)^{-1}S_1C_{10}'(0)v\;\!\P_{n'}
& \hbox{if~~$\lambda_{k,n}=\lambda=\lambda_{k,n'}$.}
\end{cases}
$$
\end{enumerate}
\end{Proposition}

Before the proof, we note that the r.h.s. of \eqref{sol1} can be rewritten as in
\cite[Sec.~3.3]{RT14}\hspace{1pt}:
\begin{align}
&\M_k(\lambda,\kappa)\nonumber\\
&=\kappa\big(I_0(\kappa)+S_0\big)^{-1}\nonumber\\
&\quad+\Big(S_0\big(I_0(\kappa)+S_0\big)^{-1}-C_{00}(\kappa)\Big)S_0
\big(I_1(\kappa)+S_1\big)^{-1}S_0
\Big(\big(I_0(\kappa)+S_0\big)^{-1}S_0+C_{00}(\kappa)\Big)\nonumber\\
&\quad+\frac1\kappa\Big\{\Big(S_1\big(I_0(\kappa)+S_0\big)^{-1}-C_{10}(\kappa)\Big)
\big(I_1(\kappa)+ S_1\big)^{-1}-\Big(S_0\big(I_0(\kappa)+S_0\big)^{-1}
-C_{00}(\kappa)\Big)C_{11}(\kappa)\Big\}\nonumber\\
&\qquad\times S_1\big(I_2(\kappa)+S_2\big)^{-1}S_1\Big\{\big(I_1(\kappa)+S_1\big)^{-1}
\Big(\big(I_0(\kappa)+S_0\big)^{-1}S_1+C_{10}(\kappa)\Big)\nonumber\\
&\qquad+C_{11}(\kappa)\Big(\big(I_0(\kappa)+S_0\big)^{-1}S_0
+C_{00}(\kappa)\Big)\Big\}\nonumber\\
&\quad+\frac1{\kappa^2}\Big\{\Big[\Big(S_2\big(I_0(\kappa)+S_0\big)^{-1}
-C_{20}(\kappa)\Big)\big(I_1(\kappa)+S_1\big)^{-1}\nonumber\\
&\qquad-\Big(S_0\big(I_0(\kappa)+S_0\big)^{-1}-C_{00}(\kappa)\Big)C_{21}(\kappa)\Big]
\big(I_2(\kappa)+S_2\big)^{-1}\nonumber\\
&\qquad-\Big[\Big(S_1\big(I_0(\kappa)+S_0\big)^{-1}-C_{10}(\kappa)\Big)
\big(I_1(\kappa)+S_1\big)^{-1}\nonumber\\
&\qquad-\Big(S_0\big(I_0(\kappa)+S_0\big)^{-1}
-C_{00}(\kappa)\Big)C_{11}(\kappa)\Big]C_{22}(\kappa)\Big\}S_2I_3(\kappa)^{-1}S_2\nonumber\\
&\qquad\times\Big\{\big(I_2(\kappa)+S_2\big)^{-1}\Big[\big(I_1(\kappa)+S_1\big)^{-1}
\Big(\big(I_0(\kappa)+S_0\big)^{-1}S_2+C_{20}(\kappa)\Big)\nonumber\\
&\qquad+C_{21}(\kappa)\Big(\big(I_0(\kappa)+S_0\big)^{-1}S_0+C_{00}(\kappa)\Big)\Big]\nonumber\\
&\qquad+C_{22}(\kappa)\Big[\big(I_1(\kappa)+S_1\big)^{-1}
\Big(\big(I_0(\kappa)+S_0\big)^{-1}S_1+C_{10}(\kappa)\Big)\nonumber\\
&\qquad+C_{11}(\kappa)\Big(\big(I_0(\kappa)+S_0\big)^{-1}S_0
+C_{00}(\kappa)\Big)\Big]\Big\}.\label{eq_grosse}
\end{align}
The interest in this formulation is that the projections $S_\ell$ (which lead to
simplifications in the proof below) have been put into evidence at the beginning or at
the end of each term.

\begin{proof}
(a) Some lengthy, but direct, computations taking into account the expansion
\eqref{eq_grosse}, the relation $\big(I_\ell(0)+S_\ell\big)^{-1}S_\ell=S_\ell$, the
expansion
\begin{equation}\label{expansion_beta}
\beta_{k,n}(\lambda-\kappa^2)^{-1}
=\beta_{k,n}(\lambda)^{-1}\left(1+\frac{\kappa^2}{4(\lambda-\lambda_{k,n})}
+\O(\kappa^4)\right),
\quad\lambda_{k,n}<\lambda,
\end{equation}
and Lemma \ref{relations_simples}(b) lead to the equality
\begin{align*}
&\lim_{\kappa\to0}
\beta_{k,n}(\lambda-\kappa^2)^{-1}\;\!\P_nv\;\!\M_k(\lambda,\kappa)
v\;\!\P_{n'}\beta_{k,n'}(\lambda-\kappa^2)^{-1}\\
&=\beta_{k,n}(\lambda)^{-1}\;\!\P_nvS_0\big(I_1(0)+S_1\big)^{-1}
S_0v\;\!\P_{n'}\beta_{k,n'}(\lambda)^{-1}\\
&\quad-\beta_{k,n}(\lambda)^{-1}\;\!\P_nv
\big(C_{20}'(0)+S_0 C_{21}'(0)\big)S_2I_3(0)^{-1}S_2
\big(C_{20}'(0)+C_{21}'(0)S_0\big)v\;\!\P_{n'}\beta_{k,n'}(\lambda)^{-1}.
\end{align*}
Moreover, Lemmas \ref{relations_simples}(a) and \ref{relations_simples}(d) imply that
\begin{align}
C_{20}(\kappa)
&=\big(I_0(\kappa)+S_0\big)^{-1}\big[v\;\!\P v+\kappa\;\!M_1(\kappa),S_2\big]
\big(I_0(\kappa)+S_0\big)^{-1}\nonumber\\
&=\kappa\;\!\big(I_0(0)+S_0\big)^{-1}\big[M_1(0),S_2\big]\big(I_0(0)+S_0\big)^{-1}
+\Oa(\kappa^2)\nonumber\\
&=\Oa(\kappa^2),\label{eq_C20}
\end{align}
and Lemma \ref{relations_simples}(d) and the expansion \eqref{eq23} imply that
\begin{align*}
C_{21}(\kappa)
&=\big(I_1(\kappa)+S_1\big)^{-1}\big[S_0M_1(0)S_0+\kappa\;\!M_2(\kappa),S_2\big]
\big(I_1(\kappa)+S_1\big)^{-1}\nonumber\\
&=\kappa\;\!\big(I_1(\kappa)+S_0\big)^{-1}\big[M_2(\kappa),S_2\big]
\big(I_1(\kappa)+S_0\big)^{-1}\nonumber\\
&=\kappa\;\!\big(I_1(\kappa)+S_0\big)^{-1}
\big[-S_0M_1(0)\big(I_0(0)+S_0\big)^{-1}M_1(0)S_0,S_2\big]
\big(I_1(\kappa)+S_0\big)^{-1}+\Oa(\kappa^2)\nonumber\\
&=\Oa(\kappa^2).
\end{align*}
Therefore, one has $C_{20}'(0)=C_{21}'(0)=0$, and thus
\begin{align*}
&\lim_{\kappa\to0}
\beta_{k,n}(\lambda-\kappa^2)^{-1}\;\!\P_nv\;\!\M_k(\lambda,\kappa)
v\;\!\P_{n'}\beta_{k,n'}(\lambda-\kappa^2)^{-1}\\
&=\beta_{k,n}(\lambda)^{-1}\;\!\P_nvS_0\big(I_1(0)+S_1\big)^{-1}
S_0v\;\!\P_{n'}\beta_{k,n'}(\lambda)^{-1}.
\end{align*}
Since
\begin{equation}\label{eq_start}
S_k(\lambda-\kappa^2)_{nn'}-\delta_{nn'}
=-2i\hspace{1pt}\beta_{k,n}(\lambda-\kappa^2)^{-1}\;\!\P_nv\;\!\M_k(\lambda,\kappa)
v\;\!\P_{n'}\beta_{k,n'}(\lambda-\kappa^2)^{-1},
\end{equation}
this proves the claim.

(b.1) We first consider the case $\lambda_{k,n}<\lambda$, $\lambda_{k,n'}=\lambda$
(the case $\lambda_{k,n}=\lambda$, $\lambda_{k,n'}<\lambda$ is not presented since it
is similar). An inspection of the expansion \eqref{eq_grosse} taking into account the
relation
$
\big(I_\ell(\kappa)+S_\ell\big)^{-1}=\big(I_\ell(0)+S_\ell\big)^{-1}+\Oa(\kappa)
$
and the relation $\big(I_\ell(0)+S_\ell\big)^{-1}S_\ell=S_\ell$ leads to the equation
\begin{align*}
&\beta_{k,n}(\lambda-\kappa^2)^{-1}\;\!\P_nv\;\!\M_k(\lambda,\kappa)
v\;\!\P_{n'}\beta_{k,n'}(\lambda-\kappa^2)^{-1}\\
&=\beta_{k,n}(\lambda-\kappa^2)^{-1}\;\!\P_nv\;\!\bigg\{
\Oa(\kappa)+S_0\big(I_1(\kappa)+S_1\big)^{-1}S_0\\
&\quad+\frac1\kappa\big(S_1+\Oa(\kappa)\big)S_1
\big(I_2(\kappa)+S_2\big)^{-1}S_1\big(S_1+\Oa(\kappa)\big)\\
&\quad+\frac1{\kappa^2}\left[\Oa(\kappa^2)+S_2\big(I_0(\kappa)+ S_0\big)^{-1}
\big(I_1(\kappa)+S_1\big)^{-1}\big(I_2(\kappa)+S_2\big)^{-1}
-C_{20}(\kappa)-S_0C_{21}(\kappa)-S_1C_{22}(\kappa)\right]\\
&\qquad\times S_2 I_3(\kappa)^{-1}S_2
\Big[\Oa(\kappa^2)+\big(I_2(\kappa)+S_2\big)^{-1}\big(I_1(\kappa)+S_1\big)^{-1}
\big(I_0(\kappa)+S_0\big)^{-1}S_2+C_{20}(\kappa)+C_{21}(\kappa)S_0\\
&\qquad+C_{22}(\kappa)S_1\Big]\bigg\}
\;\!v\;\!\P_{n'}\beta_{k,n'}(\lambda-\kappa^2)^{-1}.
\end{align*}
An application of Lemma \ref{relations_simples}(a)-(b) to the previous equation gives
\begin{align*}
&\beta_{k,n}(\lambda-\kappa^2)^{-1}\;\!\P_nv\;\!\M_k(\lambda,\kappa)
v\;\!\P_{n'}\beta_{k,n'}(\lambda-\kappa^2)^{-1}\\
&=\beta_{k,n}(\lambda-\kappa^2)^{-1}\;\!\P_nv
\left\{\Oa(\kappa)-\frac1{\kappa^2}\big(\Oa(\kappa^2)+C_{20}(\kappa)
+S_0C_{21}(\kappa)\big)S_2 I_3(\kappa)^{-1}S_2\big(\Oa(\kappa^2)
+C_{20}(\kappa)\big)\right\}\\
&\qquad\times v\;\!\P_{n'}\beta_{k,n'}(\lambda-\kappa^2)^{-1}.
\end{align*}
Finally, if one takes into account the expansion \eqref{expansion_beta} for
$\beta_{k,n}(\lambda-\kappa^2)^{-1}$ and the equality
$\beta_{k,n'}(\lambda-\kappa^2)^{-1}=|\kappa|^{-1/2}$, one ends up with
\begin{align*}
&\beta_{k,n}(\lambda-\kappa^2)^{-1}\;\!\P_nv\;\!\M_k(\lambda,\kappa)
v\;\!\P_{n'}\beta_{k,n'}(\lambda-\kappa^2)^{-1}\\
&=-\frac1{\kappa^2|\kappa|^{1/2}}\;\!\beta_{k,n}(\lambda)^{-1}\P_nv
\big(C_{20}(\kappa)+S_0C_{21}(\kappa)\big)S_2 I_3(\kappa)^{-1}S_2C_{20}(\kappa)
v\;\!\P_{n'}+\O(|\kappa|^{1/2}).
\end{align*}
Since $C_{20}(\kappa)=\Oa(\kappa^2)$ (see \eqref{eq_C20}), one infers that
$
\beta_{k,n}(\lambda-\kappa^2)^{-1}\P_nv\;\!\M_k(\lambda,\kappa)
v\;\!\P_{n'}\beta_{k,n'}(\lambda-\kappa^2)^{-1}
$
vanishes as $\kappa\to0$, and thus that the limit
$\lim_{\kappa\to0}S_k(\lambda-\kappa^2)_{nn'}$ also vanishes due to \eqref{eq_start}.

(b.2) We are left with the case $\lambda_{k,n}=\lambda=\lambda_{k,n'}$. An inspection
of the expansion \eqref{eq_grosse} taking into account the relation
$\big(I_\ell(\kappa)+S_\ell\big)^{-1}=\big(I_\ell(0)+S_\ell\big)^{-1}+\Oa(\kappa)$,
the relation $\big(I_\ell(0)+S_\ell\big)^{-1}S_\ell=S_\ell$ and Lemma
\ref{relations_simples}(a) leads to the equation
\begin{align*}
&\beta_{k,n}(\lambda-\kappa^2)^{-1}\;\!\P_nv\;\!\M_k(\lambda,\kappa)
v\;\!\P_{n'}\beta_{k,n'}(\lambda-\kappa^2)^{-1}\\
&=\beta_{k,n}(\lambda-\kappa^2)^{-1}\;\!\P_nv\;\!\bigg\{\Oa(\kappa^2)
+\kappa\big(I_0(\kappa)+S_0\big)^{-1}
-\frac1\kappa\;\!C_{10}(\kappa)S_1\big(I_2(\kappa)+S_2\big)^{-1}S_1C_{10}(\kappa)\\
&\quad-\frac1{\kappa^2}\;\!\big(\Oa(\kappa^2)+C_{20}(\kappa)\big)S_2I_3(\kappa)^{-1}
S_2\big(\Oa(\kappa^2)+C_{20}(\kappa)\big)\bigg\}v\;\!\P_{n'}
\beta_{k,n'}(\lambda-\kappa^2)^{-1}.
\end{align*}
Therefore, since
$
\beta_{k,n}(\lambda-\kappa^2)^{-1}
=\beta_{k,n'}(\lambda-\kappa^2)^{-1}
=|\kappa|^{-1/2}
$
and $C_{20}(\kappa)\in\Oa(\kappa^2)$, one obtains that
\begin{align*}
&\lim_{\kappa\to0}\beta_{k,n}(\lambda-\kappa^2)^{-1}\;\!\P_nv\;\!
\M_k(\lambda,\kappa)v\;\!\P_{n'}\beta_{k,n'}(\lambda-\kappa^2)^{-1} \\
&=-i\;\!\P_nv\big(I_0(0)+S_0\big)^{-1}v\;\!\P_{n'}
+i\;\!\P_nv\;\!C_{10}'(0)S_1\big(I_2(0)+S_2\big)^{-1}S_1C_{10}'(0)v\;\!\P_{n'},
\end{align*}
and thus that
$$
\lim_{\kappa\to0}S_k(\lambda-\kappa^2)_{nn'}
=\delta_{nn'}-2\;\!\P_nv\big(I_0(0)+S_0\big)^{-1}v\;\!\P_{n'}
+2\;\!\P_nv\;\!C_{10}'(0)S_1\big(I_2(0)+S_2\big)^{-1}S_1C_{10}'(0)v\;\!\P_{n'}
$$
due to \eqref{eq_start}.
\end{proof}

We finally consider the continuity of the scattering matrix at embedded eigenvalues
not located at thresholds.

\begin{Proposition}
Assume that $V\in\linf(\R;\R)$ is $2\pi$-periodic, take
$\lambda\in\sigma_{\rm p}(H_k^V)\setminus\tau_k$, $\kappa\in(0,\varepsilon)$ or
$i\kappa\in(0,\varepsilon)$ with $\varepsilon>0$ small enough, and let $n,n'\in\Z$.
Then, if $\lambda_{k,n},\lambda_{k,n'}<\lambda$, the limit
$\lim_{\kappa\to0}S_k(\lambda-\kappa^2)_{nn'}$ exists and is given by
\begin{equation}\label{eqSvp}
\lim_{\kappa\to0}S_k(\lambda-\kappa^2)_{nn'}
=\delta_{nn'}-2i\hspace{1pt}\beta_{k,n}(\lambda)^{-1}\P_nv\big(J_0(0)+S\big)^{-1}
v\;\!\P_{n'}\beta_{k,n'}(\lambda)^{-1}
\end{equation}
\end{Proposition}

\begin{proof}
We know from \eqref{eq_expansion_2} that
$$
\M_k(\lambda,\kappa)
=\big(J_0(\kappa)+S\big)^{-1}
+\frac1{\kappa^2}\big(J_0(\kappa)+S)^{-1}SJ_1(\kappa)^{-1}S\big(J_0(\kappa)+S\big)^{-1},
$$
with $S$ the Riesz projection associated with the value $0$ of the operator
$$
T_0=u+\sum_{m\in\Z_k(\lambda)^\bot}\frac{v\;\!\P_mv}{\beta_{k,m}(\lambda)^2}
+i\sum_{m\in\Z_k(\lambda)}\frac{v\;\!\P_mv}{\beta_{k,m}(\lambda)^2}\;\!.
$$
Now, since $J_0(\kappa)=T_0+\kappa^2T_1(\kappa)$ with $T_1(\kappa)\in \Oa(1)$, a
commutation of $S$ with $\big(J_0(\kappa)+S\big)^{-1}$ gives
$$
\M_k(\lambda,\kappa)
=\big(J_0(\kappa)+S\big)^{-1}
+\frac1{\kappa^2}\;\!\big\{S\big(J_0(\kappa)+S)^{-1}+\Oa(\kappa^2)\big\}
SJ_1(\kappa)^{-1}S\big\{\big(J_0(\kappa)+S\big)^{-1}S+\Oa(\kappa^2)\big\}.
$$
In addition, an application of \cite[Lemma~2.5]{RT14} shows that $\P_nvS=0=Sv\;\!\P_n$
for each $n\in\Z_k(\lambda)$. These relations, together with \eqref{eq_start}, imply
the equality \eqref{eqSvp}.
\end{proof}

\section{Structure of the wave operators}\label{sec5}
\setcounter{equation}{0}

In this section, we establish new stationary formulas for the wave operators
$W_{k,\pm}$ for a fixed value of $k\in[-1/2,1/2]$, and also for the full wave operators
$W_\pm(H^0,H^V)$. As before, we assume throughout the section that $V\in\linf(\R;\R)$
is $2\pi$-periodic.

We recall from \cite[Eq.~2.7.5]{Yaf92} that $W_{k,-}$ satisfies for suitable
$\varphi,\psi\in\ltwo(\Pi)$ the following equation\hspace{1pt}:
$$
\big\langle W_{k,-}\;\!\varphi,\psi\big\rangle_{\ltwo(\Pi)}
=\int_\R\d\lambda\,\lim_{\varepsilon\searrow0}\frac\varepsilon\pi
\big\langle R^0_k(\lambda- i\varepsilon)\varphi,
R^V_k(\lambda- i\varepsilon)\psi\big\rangle_{\ltwo(\Pi)}.
$$
We also recall from \cite[Sec.~1.4]{Yaf92} that if
$
\delta_\varepsilon\big(H^0_k-\lambda\big)
:=\frac{\pi^{-1}\varepsilon}{(H^0_k-\lambda)^2+\varepsilon^2}
$
for $\varepsilon>0$, then the limit
$
\lim_{\varepsilon\searrow0}
\big\langle\delta_\varepsilon\big(H^0_k-\lambda\big)\varphi,
\psi\big\rangle_{\ltwo(\Pi)}
$
exists for a.e. $\lambda\in\R$ and verifies
$$
\langle \varphi,\psi\rangle_{\ltwo(\Pi)}
=\int_\R\d\lambda\,\lim_{\varepsilon\searrow0}
\big\langle\delta_\varepsilon(H^0_k-\lambda)\varphi,\psi\big\rangle_{\ltwo(\Pi)}.
$$
So, by taking \eqref{eq_resolv} into account and by using the fact that
$
\lim_{\varepsilon\searrow0}\big\|\delta_\varepsilon\big(H^0_k-\lambda\big)\big\|
_{\B(\ltwo(\Pi))}=0
$
if $\lambda<k^2$, one infers that
$$
\big\langle\big(W_{k,-}-1\big)\varphi,\psi\big\rangle_{\ltwo(\Pi)}
=-\int_{k^2}^\infty\d\lambda\,\lim_{\varepsilon\searrow0}\big\langle
G^*M_k(\lambda+i\varepsilon)G\delta_\varepsilon(H^0_k-\lambda)\varphi,
R^0_k(\lambda-i\varepsilon)\psi\big\rangle_{\ltwo(\Pi)},
$$
with
$$
M_k(z):=\big(u+GR^0_k(z)G^*\big)^{-1},\quad z\in\C\setminus\R.
$$

Below, we derive an expression for the operator $(W_{k,-}-1)$ in the spectral
representation of $H^0_k$; that is, for the operator $\U_k(W_{k,-}-1)\U_k^*$. For that
purpose, we decompose the operator $G$ into the product $G=v\;\!\gamma_0$, with
$\gamma_0\in\B\big(\widetilde\H^1(\Pi);\ltwo(\T)\big)$ given by
$$
(\gamma_0\varphi)(\theta):=\varphi(\theta,0),
\quad\varphi\in\widetilde\H^1(\Pi),~\theta\in\T.
$$
We also define the set
\begin{align*}
\D_k
:=\Big\{\xi\in\Hrond_k\mid\xi_n=\rho_n\otimes\e^{in(\cdot)},
~\rho_n\in C^\infty_{\rm c}&\big((\lambda_{k,n},\infty)\setminus
\{\tau_k\cup\sigma_{\rm p}(H^V_k)\}\big),\\
&~\rho_n\not\equiv0\hbox{ for a finite number of $n\in\Z$}\Big\}
\end{align*}
which is dense in $\Hrond_k$ since the point spectrum of $H^V_k$ has no accumulation
point except possibly at $+\infty$. Finally, we give the short following lemma, which
will be useful in the subsequent computations for the wave operators.

\begin{Lemma}\label{lem11}
For $\xi\in\D_k$ and $\lambda\ge k^2$, one has
\begin{enumerate}
\item[(a)]
$
\gamma_0\U_k^*\xi
=\pi^{-1/2}\sum_{n\in\Z}\int_{\lambda_{k,n}}^\infty\d\mu\,
\beta_{k,n}(\mu)^{-1}\xi_n(\mu)\in\ltwo(\T)
$,
\item[(b)]
$
\slim_{\varepsilon\searrow0}\gamma_0\U_k^*\delta_\varepsilon(L_k-\lambda)\xi
=\pi^{-1/2}\sum_{n\in\Z_k(\lambda)}\beta_{k,n}(\lambda)^{-1}\xi_n(\lambda)\in\ltwo(\T)
$.
\end{enumerate}
\end{Lemma}

\begin{proof}
The equality in (a) follows from a direct computation, and the inclusion in
$\ltwo(\T)$ follows from the fact that $\U_k^*\xi\in\widetilde\H^1(\Pi)$. For (b), it
is sufficient to note that the map $\mu\mapsto\beta_{k,n}(\mu)^{-1}\xi_n(\mu)$ extends
trivially to a continuous function on $\R$ with compact support in
$(\lambda_{k,n},\infty)$, and then to use the convergence of the Dirac delta sequence
$\delta_\varepsilon(\,\cdot\,-\lambda)$.
\end{proof}

Thus, if we let $\xi,\zeta\in\D_k$ and take the previous observations into account, we
obtain the equalities
\begin{align}
&\big\langle\U_k\big(W_{k,-}-1\big)\U_k^*\xi,\zeta\big\rangle_{\Hrond_k}\nonumber\\
&=-\int_{k^2}^\infty\d\lambda\,\lim_{\varepsilon\searrow0}
\big\langle\gamma_0^*v\;\! M_k(\lambda+i\varepsilon)v\;\!\gamma_0\U_k^*
\delta_\varepsilon(L_k-\lambda)\xi,\U_k^*(L_k-\lambda+i\varepsilon)^{-1}
\zeta\big\rangle_{\ltwo(\Pi)}\nonumber\\
&=-\int_{k^2}^\infty\d\lambda\,\lim_{\varepsilon\searrow0}
\big\langle v\;\!M_k(\lambda+i\varepsilon)v\;\!\gamma_0\U_k^*
\delta_\varepsilon(L_k-\lambda)\xi,\gamma_0\U_k^*
(L_k-\lambda+i\varepsilon)^{-1}\zeta\big\rangle_{\ltwo(\T)}\nonumber\\
&=-\int_{k^2}^\infty\d\lambda\,\lim_{\varepsilon\searrow0}
\left\langle\pi^{-1/2}v\;\! M_k(\lambda+i\varepsilon)v\;\!\gamma_0\U_k^*
\delta_\varepsilon(L_k-\lambda)\xi,
\sum_{n\in\Z}\int_{\lambda_{k,n}}^\infty\d\mu\,\frac{\beta_{k,n}(\mu)^{-1}}{\mu-\lambda+i\varepsilon}
\;\!\zeta_n(\mu)\right\rangle_{\ltwo(\T)}\nonumber\\
&=-\sum_{n\in\Z}\int_{\lambda_{k,n}}^\infty\d\lambda\,\lim_{\varepsilon\searrow0}
\left\langle\pi^{-1/2}v
 M_k(\lambda+i\varepsilon)v\;\!\gamma_0\U_k^*\delta_\varepsilon(L_k-\lambda)\xi,
\int_{\lambda_{k,n}}^\infty\d\mu\,\frac{\beta_{k,n}(\mu)^{-1}}
{\mu-\lambda+i\varepsilon}\;\!\zeta_n(\mu)\right\rangle_{\ltwo(\T)}
\label{leadingterm}\\
&\quad-\sum_{n\in\Z}\int_{k^2}^{\lambda_{k,n}}\d\lambda\,\lim_{\varepsilon\searrow0}
\left\langle\pi^{-1/2}v\;\! M_k(\lambda+i\varepsilon)v\;\!\gamma_0\U_k^*
\delta_\varepsilon(L_k-\lambda)\xi,
\int_{\lambda_{k,n}}^\infty\d\mu\,\frac{\beta_{k,n}(\mu)^{-1}}{\mu-\lambda+i\varepsilon}
\;\!\zeta_n(\mu)\right\rangle_{\ltwo(\T)}\label{remainderterm}
\end{align}
with the sums over $n$ being finite. In the next two sections, we study separately the
terms \eqref{leadingterm} and \eqref{remainderterm}.

\subsection{Wave operators\hspace{1pt}: the leading term}\label{section_leading}

We prove in this section an explicit formula for the term \eqref{leadingterm} in the
expression for $(W_{k,-}-1)$ in terms of the generator of dilations in $\R_+$. For
this, we recall that the dilation group $\{U^+_\tau\}_{\tau\in\R}$ in $\ltwo(\R_+)$ is
defined by
$$
\big(U^+_\tau\varphi\big)(\lambda):=\e^{\tau/2}\varphi(\e^\tau\lambda),
\quad \varphi\in C_{\rm c}(\R_+),~\lambda\in\R_+,~\tau\in\R,
$$
and that the self-adjoint generator of $\{U^+_\tau\}_{\tau\in\R}$ is denoted by $A_+$.

\begin{Proposition}\label{prop_leading}
Assume that $V\in\linf(\R;\R)$ is $2\pi$-periodic and take $\xi,\zeta\in\D_k$. Then,
we have
\begin{align*}
&-\sum_{n\in\Z}\int_{\lambda_{k,n}}^\infty\d\lambda\,\lim_{\varepsilon\searrow0}
\left\langle\pi^{-1/2}v\;\! M_k(\lambda+i\varepsilon)v\;\!\gamma_0\U_k^*
\delta_\varepsilon(L_k-\lambda)\xi,\int_{\lambda_{k,n}}^\infty\d\mu\,
\frac{\beta_{k,n}(\mu)^{-1}}{\mu-\lambda+i\varepsilon}\;\!\zeta_n(\mu)
\right\rangle_{\ltwo(\T)}\\
&=\big\langle\U_k\big(1\otimes R(A_+)\big)(S_k-1)\U_k^*\xi,
\zeta\big\rangle_{\Hrond_k}
\end{align*}
with
\begin{equation}\label{def_R}
R(x):=\frac12\big(1+\tanh(\pi x)+i\cosh(\pi x)^{-1}\big),\quad x\in\R.
\end{equation}
\end{Proposition}

\begin{proof}
(i) Take $\eta\in C^\infty_{\rm c}(\R_+)$ and $x\in\R_+$, let $\F$ be the Fourier
transform on $\R$, and write $\chi_+$ for the characteristic function for $\R_+$.
Then, we have
\begin{align*}
(\Theta\eta)(x)
&:=2\int_0^\infty\d y\int_0^\infty\d z\,x\e^{i(y^2-x^2)z}\eta(y)\\
&=2^{3/2}\pi^{1/2}\int_0^\infty\d y\,\big(\F^*\chi_+\big)(y^2-x^2)\;\!x\;\!\eta(y)\\
&=2^{3/2}\pi^{1/2}\int_\R\d z\,\big(\F^*\chi_+\big)\big(x^2(\e^{2z}-1)\big)
\;\!x^2\e^z\eta(\e^zx)\quad(y=\e^zx)\\
&=2^{3/2}\pi^{1/2}\int_\R\d z\,\big(\F^*\chi_+\big)\big(x^2(\e^{2z}-1)\big)
\;\!x^2\e^{z/2}\big(U^+_z\eta\big)(x).
\end{align*}
Then, by using the fact that
$
\F^*\chi_+=2^{-1/2}\pi^{1/2}\;\!\delta_0
+i\hspace{1pt}(2\pi)^{-1/2}\;\!\Pv\frac1{(\,\cdot\,)}
$
with $\delta_0$ the Dirac delta distribution and $\Pv$ the principal value, one gets
that
\begin{align*}
(\Theta\eta)(x)
&=2\int_\R\d z\left(\pi\;\!\delta_0(\e^{2z}-1)
+i\;\!\Pv\frac{\e^{z/2}}{\e^{2z}-1}\right)\big(U^+_z\eta\big)(x)\\
&=\int_\R\d z\left(\pi\;\!\delta_0(z)
+\frac i2\;\!\Pv\left(\frac1{\sinh(z/2)}-\frac1{\cosh(z/2)}\right)\right)
\big(U^+_z\eta\big)(x).
\end{align*}
So, by taking into account the equality \cite[Table 20.1]{Jef95}
\begin{align*}
(2\pi)^{1/2}\big(\F\overline R\big)(z)
&=\pi\;\!\delta_0(-z)
+\frac i2\;\!\Pv\left(\frac1{\sinh(-z/2)}-\frac1{\cosh(-z/2)}\right)
\end{align*}
with $R$ as in \eqref{def_R}, one infers that
$$
(\Theta\eta)(x)
=(2\pi)^{1/2}\int_\R\d z\,\big(\F\overline R\big)(-z)\big(U^+_z\eta\big)(x)
=2\pi\big(\overline R(-A_+)\eta\big)(x).
$$
Therefore, one has for each $\zeta\in\D_k$, $n\in\Z$ and $\lambda>\lambda_{k,n}$ the
following equalities in $\ltwo(\T)$\hspace{1pt}:
\begin{align}
&2\pi\big(\U_k(1\otimes\overline R(A_+))\U_k^*\zeta\big)_n(\lambda)\nonumber\\
&=2\pi\big(\U_k(1\otimes\Fc^*\overline R(-A_+)\Fc)\U_k^*\zeta\big)_n(\lambda)\nonumber\\
&=\big(\U_k(1\otimes\Fc^*\Theta\Fc)\U_k^*\zeta\big)_n(\lambda)\nonumber\\
&=2^{-1/2}(\lambda-\lambda_{k,n})^{-1/4}
\big((\P_n\otimes\Theta\Fc)\U_k^*\zeta\big)
\big(\,\cdot\,,(\lambda-\lambda_{k,n})^{1/2}\big)\nonumber\\
&=2^{1/2}(\lambda-\lambda_{k,n})^{-1/4}\int_0^\infty\d y\int_0^\infty\d z\,
(\lambda-\lambda_{k,n})^{1/2}\e^{i(y^2-\lambda+\lambda_{k,n})z}
\big((\P_n\otimes\Fc)\U_k^*\zeta\big)(\,\cdot\,,y)\nonumber\\
&=2\;\!(\lambda-\lambda_{k,n})^{1/4}\int_0^\infty\d y\int_0^\infty\d z\,
\e^{i(y^2-\lambda+\lambda_{k,n})z}y^{1/2}\zeta_n(y^2+\lambda_{k,n})\nonumber\\
&=\int_{\lambda_{k,n}}^\infty\d\mu\int_0^\infty\d z\,
\e^{i(\mu-\lambda)z}\beta_{k,n}(\lambda)\beta_{k,n}(\mu)^{-1}\;\!\zeta_n(\mu).
\label{eq_RA}
\end{align}

(ii) Let $\xi,\zeta\in\D_k$ and take $\varepsilon>0$, $n\in\Z$ and
$\lambda\in[\lambda_{k,n},\infty)\setminus\{\tau_k\cup\sigma_{\rm p}(H_k^V)\}$. Then,
Lemma \ref{lem11}(a), the formula
$
(\mu-\lambda+i\varepsilon)^{-1}
=-i\int_0^\infty\d z\,\e^{i(\mu-\lambda)z}\e^{-\varepsilon z}
$
and Fubini's theorem imply that
\begin{align}
&\lim_{\varepsilon\searrow0}\left\langle\pi^{-1/2}v\;\! M_k(\lambda+i\varepsilon)v\;\!
\gamma_0\U_k^*\delta_\varepsilon(L_k-\lambda)\xi,
\int_{\lambda_{k,n}}^\infty\d\mu\,\frac{\beta_{k,n}(\mu)^{-1}}
{\mu-\lambda+i\varepsilon}\;\!\zeta_n(\mu)\right\rangle_{\ltwo(\T)}\nonumber\\
&=\lim_{\varepsilon\searrow0}\left\langle\pi^{-1/2}\beta_{k,n}(\lambda)^{-1}\P_{n}v
\;\! M_k(\lambda+i\varepsilon)v\;\!\gamma_0\U_k^*\delta_\varepsilon(L_k-\lambda)\xi,
\int_{\lambda_{k,n}}^\infty\d\mu\,\frac{\beta_{k,n}(\lambda)\beta_{k,n}(\mu)^{-1}}
{\mu-\lambda+i\varepsilon}\;\!\zeta_n(\mu)\right\rangle_{\ltwo(\T)}\nonumber\\
&=\lim_{\varepsilon\searrow0}\int_0^\infty\d z\,\e^{-\varepsilon z}
\left\langle g_\varepsilon(n,\lambda),
\int_{\lambda_{k,n}}^\infty\d\mu\,\e^{i(\mu-\lambda)z}
\beta_{k,n}(\lambda)\beta_{k,n}(\mu)^{-1}\;\!\zeta_n(\mu)\right\rangle_{\ltwo(\T)}
\label{integrant_eps}
\end{align}
with
$$
g_\varepsilon(n,\lambda)
:=i\;\!\pi^{-1}\beta_{k,n}(\lambda)^{-1}\P_{n}v\;\!
 M_k(\lambda+i\varepsilon)v\sum_{n'\in\Z}\int_{\lambda_{k,n'}}^\infty\d\nu\,
\beta_{k,n'}(\nu)^{-1}\delta_\varepsilon(\nu-\lambda)\xi_{n'}(\nu).
$$
Now, we already know from \eqref{def_M0} that
$\lim_{\varepsilon\searrow0}M_k(\lambda+i\varepsilon)=\M_k(\lambda,0)$ in
$\B\big(\ltwo(\T)\big)$ and we have
$$
\slim_{\varepsilon\searrow0}\sum_{n'\in\Z}\int_{\lambda_{k,n'}}^\infty\d\nu\,
\beta_{k,n'}(\nu)^{-1}\delta_\varepsilon(\nu-\lambda)\xi_{n'}(\nu)
=\sum_{n'\in\Z_k(\lambda)}\beta_{k,n'}(\lambda)^{-1}\xi_{n'}(\lambda)
$$
in $\ltwo(\T)$. Therefore, we have
$$
g_0(n,\lambda)
:=\slim_{\varepsilon\searrow0}g_{\varepsilon}(n,\lambda)
=i\pi^{-1}\beta_{k,n}(\lambda)^{-1}\P_{n}v\;\!
\M_k(\lambda,0)v\sum_{n'\in\Z_k(\lambda)}\beta_{k,n'}(\lambda)^{-1}\xi_{n'}(\lambda)
$$
in $\ltwo(\T)$, and the integrant in \eqref{integrant_eps} can be bounded independently of
$\varepsilon\in(0,1)$\hspace{1pt}:
\begin{align}
&\left|\;\!\e^{-\varepsilon z}
\left\langle g_\varepsilon(n,\lambda),
\int_{\lambda_{k,n}}^\infty\d\mu\,\e^{i(\mu-\lambda)z}
\beta_{k,n}(\lambda)\beta_{k,n}(\mu)^{-1}\;\!\zeta_n(\mu)\right\rangle_{\ltwo(\T)}\right|
\nonumber\\
&\le{\rm Const.}
\left\|\int_{\lambda_{k,n}}^\infty\d\mu\,\e^{i(\mu-\lambda)z}
\beta_{k,n}(\lambda)\beta_{k,n}(\mu)^{-1}\;\!\zeta_n(\mu)\right\|_{\ltwo(\T)}.
\label{bounded_eps}
\end{align}

In order to exchange the integral over $z$ and the limit $\varepsilon\searrow0$ in
\eqref{integrant_eps}, it remains to show that the r.h.s. of \eqref{bounded_eps}
belongs to $\lone(\R_+,\d z)$. For that purpose, we note that
\begin{align*}
\left\|\int_{\lambda_{k,n}}^\infty\d\mu\,\e^{i(\mu-\lambda)z}
\beta_{k,n}(\lambda)\beta_{k,n}(\mu)^{-1}\;\!\zeta_n(\mu)\right\|_{\ltwo(\T)}
&=\left\|\int_{\lambda_{k,n}-\lambda}^\infty\d\nu\,\e^{i\nu z}
\beta_{k,n}(\lambda)\beta_{k,n}(\nu+\lambda)^{-1}\;\!\zeta_n(\nu+\lambda)\right\|_{\ltwo(\T)}\\
&=\big\|(\F^*h_{n,\lambda})(z)\big\|_{\ltwo(\T)}
\end{align*}
with $h_{n,\lambda}$ the trivial extension of the function
$$
(\lambda_{k,n}-\lambda,\infty)\ni\nu\mapsto
(2\pi)^{1/2}\beta_{k,n}(\lambda)\beta_{k,n}(\nu+\lambda)^{-1}\;\!\zeta_n(\nu+\lambda)
\in\ltwo(\T)
$$
to all of $\R$. Then, writing $P$ for the self-adjoint operator $-i\;\!\nabla$ on $\R$,
and using the fact that
$$
h_{n,\lambda}(\nu)=
\begin{cases}
(2\pi)^{1/2}\beta_{k,n}(\lambda)\beta_{k,n}(\nu+\lambda)^{-1}\;\!\rho_n(\nu+\lambda)
\e^{in(\,\cdot\,)}
& \hbox{if}~~\nu>\lambda_{k,n}-\lambda\\
0 & \hbox{if}~~\nu\le\lambda_{k,n}-\lambda
\end{cases}
$$
with
$
\rho_n\in C^\infty_{\rm c}\big((\lambda_{k,n},\infty)\setminus
\{\tau_k\cup\sigma_{\rm p}(H^V_k)\}\big)
$,
one obtains that
$$
\big\|\big(\F^*h_{n,\lambda}\big)(z)\big\|_{\ltwo(\T)}
=\langle z\rangle^{-2}
\big\|\big(\F^*\langle P\rangle^2h_{n,\lambda}\big)(z)\big\|_{\ltwo(\T)}
\le{\rm Const.}\;\!\langle z\rangle^{-2},
\quad z\in\R_+.
$$
As a consequence, one can apply Lebesgue dominated convergence theorem and Fubini's
theorem to infer that \eqref{integrant_eps} is equal to
$$
\left\langle g_0(n,\lambda),\int_{\lambda_{k,n}}^\infty\d\mu\int_0^\infty\d z\,
\e^{i(\mu-\lambda)z}\beta_{k,n}(\lambda)\beta_{k,n}(\mu)^{-1}\;\!\zeta_n(\mu)
\right\rangle_{\ltwo(\T)}.
$$
This, together with \eqref{formule_S_matrix} and \eqref{eq_RA}, implies that
\begin{align*}
&\lim_{\varepsilon\searrow0}\left\langle\pi^{-1/2}v\;\! M_k(\lambda+i\varepsilon)v\;\!
\gamma_0\U_k^*\delta_\varepsilon(L_k-\lambda)\xi,
\int_{\lambda_{k,n}}^\infty\d\mu\,\frac{\beta_{k,n}(\mu)^{-1}}
{\mu-\lambda+i\varepsilon}\;\!\zeta_n(\mu)\right\rangle_{\ltwo(\T)}\\
&=\left\langle\sum_{n'\in\Z_k(\lambda)}2i\hspace{1pt}\beta_{k,n}(\lambda)^{-1}
\P_{n}v\;\!\M_k(\lambda,0)v\;\!\P_{n'}\beta_{k,n'}(\lambda)^{-1}\xi_{n'}(\lambda),
\big(\U_k(1\otimes\overline R(A_+))\U_k^*\zeta\big)_n(\lambda)
\right\rangle_{\ltwo(\T)}\\
&=-\big\langle\big(\U_k(S_k-1)\U_k^*\xi\big)_n(\lambda),
\big(\U_k(1\otimes\overline R(A_+))\U_k^*\zeta\big)_n(\lambda)
\big\rangle_{\ltwo(\T)}.
\end{align*}
Now, the last equality holds not only for
$
\lambda\in[\lambda_{k,n},\infty)\setminus\{\tau_k\cup\sigma_{\rm p}(H_k^V)\}
$
but for all $\lambda\in[\lambda_{k,n},\infty)$, since for each $n\in\Z$ and all
$\lambda\in\tau_k\cup\sigma_{\rm p}(H_k^V)$ we have $\xi_n(\lambda)=0$. So, we finally
obtain that
\begin{align*}
&-\sum_{n\in\Z}\int_{\lambda_{k,n}}^\infty\d\lambda\,\lim_{\varepsilon\searrow0}
\left\langle\pi^{-1/2}v\;\! M_k(\lambda+i\varepsilon)v\;\!\gamma_0\U_k^*
\delta_\varepsilon(L_k-\lambda)\xi,\int_{\lambda_{k,n}}^\infty\d\mu\,
\frac{\beta_{k,n}(\mu)^{-1}}{\mu-\lambda+i\varepsilon}\;\!\zeta_n(\mu)
\right\rangle_{\ltwo(\T)}\\
&=\left\langle\U_k\big(1\otimes R(A_+)\big)(S_k-1)\U_k^*\xi,\zeta
\right\rangle_{\Hrond_k},
\end{align*}
as desired.
\end{proof}

\subsection{Wave operators\hspace{1pt}: the remainder term}\label{section_remainder}

We prove in this section that the remaining term \eqref{remainderterm} in the
expression for $(W_{k,-}-1)$ can be written as a matrix operator in $\Hrond_k$ with
Hilbert-Schmidt components. For this, we start with a lemma which complements the
continuity properties obtained Section \ref{seccont}.

\begin{Lemma}\label{lemcont}
Assume that $V\in\linf(\R;\R)$ is $2\pi$-periodic, and choose $n,n'\in\Z$ such that
$\lambda_{k,n'}<\lambda_{k,n}$. Then, the function
\begin{equation}\label{defmap}
[\lambda_{k,n'},\lambda_{k,n}]\setminus\{\tau_k \cup \sigma_{\rm p}(H^V_k)\}
\ni\lambda\mapsto\beta_{k,n}(\lambda)^{-2}\;\!\P_nv\;\!\M_k(\lambda,0)v\;\!\P_{n'}
\in\B\big(\ltwo(\T)\big)
\end{equation}
extends to a continuous function on $[\lambda_{k,n'},\lambda_{k,n}]$.
\end{Lemma}

\begin{proof}
Since the function \eqref{defmap} is continuous on
$[\lambda_{k,n'},\lambda_{k,n}]\setminus\{\tau_k \cup \sigma_{\rm p}(H^V_k)\}$, one
only has to check that the function admits limits in $\B\big(\ltwo(\T)\big)$ as
$\lambda\to\lambda_0\in\{\tau_k\cup\sigma_{\rm p}(H^V_k)\}$.
However, in order to be able to use the asymptotic expansions of Proposition
\ref{Prop_Asymp}, we slightly change the point of view by considering values
$\lambda-\kappa^2\in\C$ with $\lambda\in\{\tau_k\cup\sigma_{\rm p}(H^V_k)\}$ and
$\kappa\to0$ in a suitable domain of $\C$ of diameter $\varepsilon>0$. Namely, we
consider the three following possible cases: when $\lambda=\lambda_{k,n'}$ and
$i\kappa\in(0,\varepsilon)$ (case 1), when $\lambda=\lambda_{k,n}$ and
$\kappa\in(0,\varepsilon)$ (case 2), and when
$\lambda\in(\lambda_{k,n'},\lambda_{k,n})\cap\{\tau_k\cup\sigma_{\rm p}(H^V_k)\}$ and
$\kappa\in(0,\varepsilon)$ or $i\kappa\in(0,\varepsilon)$ (case 3). In each case, we
choose $\varepsilon>0$ small enough so that
$\{z\in\C\mid|z|<\varepsilon\}\cap\{\tau_k\cup\sigma_{\rm p}(H^V_k)\}=\{\lambda\}$
(this is possible because $\tau_k$ is discrete and $\sigma_{\rm p}(H^V_k)$ has no
accumulation point).

(i) First, assume that $\lambda\in \sigma_{\rm p}(H^V_k)\setminus \tau_k$ and let
$\kappa\in(0,\varepsilon)$ or $i\kappa\in(0,\varepsilon)$ with $\varepsilon>0$ small
enough. Then, we know from \eqref{eq_expansion_2} that
$$
\P_nv\;\!\M_k(\lambda,\kappa)v\;\!\P_{n'}
=\P_nv\big(J_0(\kappa)+S\big)^{-1}v\;\!\P_{n'}
+\frac1{\kappa^2}\;\!\P_nv\big(J_0(\kappa)+S)^{-1}SJ_1(\kappa)^{-1}S
\big(J_0(\kappa)+S\big)^{-1}v\;\!\P_{n'}
$$
with $S$, $J_0(\kappa)$ and $J_1(\kappa)$ as in point (ii) of the proof of Proposition
\ref{Prop_Asymp}. Furthermore, point (ii) of the proof of Proposition \ref{Prop_Asymp}
implies that $[S,J_0(\kappa)]\in \Oa(\kappa^2)$, and Lemma \ref{relations_simples}(b)
(applied with $S$ instead of $S_1$) implies that $Sv\;\!\P_{n'}=0$. Therefore,
\begin{align*}
\P_nv\;\!\M_k(\lambda,\kappa)v\;\!\P_{n'}
&=\Oa(1)+\frac1{\kappa^2}\;\!\P_nv\big(J_0(\kappa)+S)^{-1}SJ_1(\kappa)^{-1}S
\;\!\big\{\big(J_0(\kappa)+S)^{-1}S+\Oa(\kappa^2)\big\}v\;\!\P_{n'}\\
&=\Oa(1).
\end{align*}
Since
$
\lim_{\kappa\to0}\beta_{k,n}(\lambda-\kappa^2)^{-2}
=|\lambda-\lambda_{k,n}|^{-1/2}
<\infty
$
for each $\lambda\in\sigma_{\rm p}(H^V_k)\setminus\tau_k$, we thus infer that the
function \eqref{defmap} (with $\lambda$ replaced by $\lambda-\kappa^2$) admits a limit
in $\B\big(\ltwo(\T)\big)$ as $\kappa\to0$.

(ii) Now, assume that $\lambda\in[\lambda_{k,n'},\lambda_{k,n}]\cap\tau_k$, and
consider the three above cases simultaneously. For this, we recall that
$i\kappa\in(0,\varepsilon)$ in case 1, $\kappa\in(0,\varepsilon)$ in case 2, and
$\kappa\in(0,\varepsilon)$ or $i\kappa\in(0,\varepsilon)$ in case 3. Also, we note
that the factor $\beta_{k,n}(\lambda-\kappa^2)^{-2}$ does not play any role in cases 1
and 3, but gives a singularity of order $|\kappa|^{-1}$ in case 2.

In the expansion \eqref{eq_grosse}, the first term (the one linear in $\kappa$) admits
a limit in $\B\big(\ltwo(\T)\big)$ as $\kappa\to0$, even in case 2. For the second
term (the one of order $\Oa(1)$ in $\kappa$) only case 2 requires a special attention:
in this case, the existence of the limit as $\kappa\to0$ follows from the inclusion
$C_{00}(\kappa)\in\Oa(\kappa)$ and the equality $\P_nvS_0=0$, which holds by Lemma
\ref{relations_simples}(a). For the third term (the one with prefactor
$\frac1{\kappa}$), in cases 1 and 3, it is sufficient to observe that
$C_{00}(\kappa),C_{10}(\kappa)\in\Oa(\kappa)$ and that $S_1v\;\!\P_{n'}=0$ by Lemma
\ref{relations_simples}(b). On the other hand, for case 2, one must take into account
the inclusions $C_{00}(\kappa),C_{10}(\kappa)\in\Oa(\kappa)$, the equality
$S_1v\;\!\P_{n'}=0$ of Lemma \ref{relations_simples}(b) and the equality
$\P_nv\;\!S_1=0$ of Lemma \ref{relations_simples}(a). For the fourth term (the one
with prefactor $\frac1{\kappa^2}$), in cases 1 and 3, it is sufficient to observe that
$C_{20}(\kappa),C_{21}(\kappa)\in\Oa(\kappa^2)$ and that
$S_2v\;\!\P_{n'}=0=S_1 v\;\!\P_{n'}$. On the other hand, for case 2, one must take
into account the inclusions $C_{20}(\kappa),C_{21}(\kappa)\in\Oa(\kappa^2)$, the
equalities $S_2v\;\!\P_{n'}=0=S_1 v\;\!\P_{n'}$, and the equality $\P_n v S_2=0$.
\end{proof}

Now, to obtain the desired formula for the term \eqref{remainderterm}, we define for
$\varepsilon>0$, $n\in\Z$, $\lambda\in\R$ and $\xi\in\D_k$ the vector
$$
g_\varepsilon(n,\lambda)
:=\pi^{-1/2}\;\!\P_n v\;\!M_k(\lambda+i\varepsilon)v\;\!\gamma_0\U_k^*
\delta_\varepsilon(L_k-\lambda)\xi\in\ltwo(\T),
$$
and we note from Proposition \ref{Prop_Asymp} and Lemma \ref{lem11}(b) that for each
$\lambda\in(k^2,\lambda_{k,n})\setminus\{\tau_k\cup\sigma_{\rm p}(H_k^V)\}$ we have
$$
g_0(n,\lambda)
:=\slim_{\varepsilon\searrow 0}g_\varepsilon(n,\lambda)
=\pi^{-1}\;\!\P_nv\;\!\M(\lambda,0)v\sum_{n'\in\Z_k(\lambda)}
\beta_{k,n'}(\lambda)^{-1}\xi_{n'}(\lambda).
$$
Then, we observe that \eqref{remainderterm} can be written as
\begin{equation}\label{termterm}
-\sum_{n\in\Z}\int_{k^2}^{\lambda_{k,n}}\d\lambda\,\lim_{\varepsilon\searrow0}
\left\langle g_\varepsilon(n,\lambda),
\int_{\lambda_{k,n}}^\infty\d\mu\,\frac{\beta_{k,n}(\mu)^{-1}}{\mu-\lambda+i\varepsilon}
\;\!\zeta_n(\mu)\right\rangle_{\ltwo(\T)}
\end{equation}
with
$$
\left|\left\langle g_\varepsilon(n,\lambda),\int_{\lambda_{k,n}}^\infty\d\mu\,
\frac{\beta_{k,n}(\mu)^{-1}}{\mu-\lambda+i\varepsilon}
\;\!\zeta_n(\mu)\right\rangle_{\ltwo(\T)}\right|
\le{\rm Const.}\left\|\int_{\lambda_{k,n}}^\infty\d\mu\,
\frac{\beta_{k,n}(\mu)^{-1}}{\mu-\lambda+i\varepsilon}
\;\!\zeta_n(\mu) \right\|_{\ltwo(\T)}.
$$
Since the r.h.s. can be bounded independently of $\varepsilon$, we infer from
Lebesgue dominated convergence theorem that \eqref{termterm} can be rewritten as
\begin{equation}\label{horreur1}
-\sum_{n\in\Z}\sum_{n'\in\Z_k(\lambda_{k,n})}\int_{\lambda_{k,n'}}^{\lambda_{k,n}}\d\lambda
\left\langle B_{nn'}(\lambda)\xi_{n'}(\lambda),\int_{\lambda_{k,n}}^\infty\d\mu\,
\frac{\beta_{k,n}(\lambda)^2\beta_{k,n}(\mu)^{-1}\beta_{k,n'}(\lambda)^{-1}}{\pi(\mu-\lambda)}
\;\!\zeta_n(\mu)\right\rangle_{\ltwo(\T)}
\end{equation}
with
\begin{equation}\label{def_Bn}
B_{nn'}(\lambda)
=\beta_{k,n}(\lambda)^{-2}\;\!\P_nv\;\!\M_k(\lambda,0)v\;\!\P_{n'}\in\B\big(\ltwo(\T)\big)
\quad\hbox{for a.e. $\lambda\in (\lambda_{k,n'},\lambda_{k,n})$.}
\end{equation}
But the map $\lambda\mapsto B_{nn'}(\lambda)$ coincides with the map \eqref{defmap}.
Therefore, Lemma \ref{lemcont} and Fubini's theorem imply that \eqref{horreur1} can be
written as $\big\langle Q_k\hspace{1pt}\xi,\zeta\big\rangle_{\Hrond_k}$, with
$Q_k:\Hrond_k\to\Hrond_k$ given for $\xi\in\D_k$, $n\in\Z$ and $\mu>\lambda_{k,n}$ by
$$
(Q_k\xi)_n(\mu)
:=-\sum_{n'\in\Z_k(\lambda_{k,n})}\int_{\lambda_{k,n'}}^{\lambda_{k,n}}
\d\lambda\,\frac{\beta_{k,n}(\lambda)^2\beta_{k,n}(\mu)^{-1}\beta_{k,n'}(\lambda)^{-1}}
{\pi(\mu-\lambda)}\;\!B_{nn'}(\lambda)\xi_{n'}(\lambda).
$$
To simplify the last formula, we define the operator
$
B_{nn'}\in\B\big(\Hrond_{k,n'};
\ltwo\big((\lambda_{k,n'},\lambda_{k,n});\P_n\;\!\ltwo(\T)\big)\big)
$
by
$$
\big(B_{nn'}\xi_{n'}\big)(\lambda):=B_{nn'}(\lambda)\xi_{n'}(\lambda)
\quad\hbox{for a.e. $\lambda\in(\lambda_{k,n'},\lambda_{k,n})$.}
$$
We also define the integral operator
$
C_{nn'}:\ltwo\big((\lambda_{k,n'},\lambda_{k,n});\P_n\;\!\ltwo(\T)\big)
\to\Hrond_{k,n}
$
with kernel
$$
C_{nn'}(\mu,\lambda)
:=\frac{\beta_{k,n}(\lambda)^2\beta_{k,n}(\mu)^{-1}\beta_{k,n'}(\lambda)^{-1}}
{\pi(\mu-\lambda)}\;\!,
\quad\mu>\lambda_{k,n},~\lambda\in (\lambda_{k,n'},\lambda_{k,n}),
$$
and show that $C_{nn'}$ is a Hilbert-Schmidt operator.

\begin{Lemma}
The operator $C_{nn'}$ is a Hilbert-Schmidt operator from
$\ltwo\big((\lambda_{k,n'},\lambda_{k,n});\P_n\;\!\ltwo(\T)\big)$ to $\Hrond_{k,n}$.
\end{Lemma}

\begin{proof}
Using the changes of variables $x:=(\mu-\lambda_{k,n})^{1/2}$,
$y:=(\lambda_{k,n}-\lambda)^{1/2}$, and the notation
$\alpha:=(\lambda_{k,n}-\lambda_{k,n'})^{1/2}$, one obtains that
\begin{align*}
&\int_{\lambda_{k,n}}^\infty\d\mu\int_{\lambda_{k,n'}}^{\lambda_{k,n}}\d\lambda\,
\frac{|\lambda-\lambda_{k,n}|\;\!|\mu-\lambda_{k,n}|^{-1/2}\;\!
|\lambda-\lambda_{k,n'}|^{-1/2}}{\pi^2(\mu-\lambda)^2}\\
&=\frac4{\pi^2}\int_0^\infty\d x\int_0^{\alpha}\d y\,
\frac{y^3(\alpha^2-y^2)^{-1/2}}{(x^2+y^2)^2}\\
&=\frac4{\pi^2}\int_0^\alpha\d y\,\frac{y^3}{(\alpha^2-y^2)^{1/2}}
\left(\frac x{2y^2(x^2+y^2)}+\frac{\arctan(x/y)}{2y^3}\right)\bigg|_{x=0}^{x=\infty}\\
&=\frac1\pi\int_0^\alpha\d y\,(\alpha^2-y^2)^{-1/2} \\
&=1/2.
\end{align*}
It follows that $C_{nn'}$ is a Hilbert-Schmidt operator from
$\ltwo\big((\lambda_{k,n'},\lambda_{k,n});\P_n\;\!\ltwo(\T)\big)$ to $\Hrond_{k,n}$
with Hilbert-Schmidt norm equal to $1/\sqrt2$.
\end{proof}

Therefore, the remainder term \eqref{remainderterm} in the expression for
$(W_{k,-}-1)$ can be written as
$\big\langle Q_k\hspace{1pt}\xi,\zeta\big\rangle_{\Hrond_k}$ with
$Q_k:\Hrond_k\to\Hrond_k$ given by
\begin{equation}\label{def_Q_k}
(Q_k\xi)_n=-\sum_{n'\in \Z_k(\lambda_{k,n})}C_{nn'}B_{nn'}\xi_{n'},
\quad\xi\in\D_k,~n\in\Z,
\end{equation}
where each summand $C_{nn'}B_{nn'}:\Hrond_{k,n'}\to\Hrond_{k,n}$ belongs to the
Hilbert-Schmidt class.

We close the section with two observations which show that the remainder term $Q_k$ is
always small in some suitable sense. First, we consider the case of a constant function $V:$

\begin{Remark}\label{rem_constant}
If the function $V$ is constant, then the remainder term $Q_k$ vanishes. Indeed, in
such a case one can easily check that the operator $\M_k(\lambda,0)$ is diagonal in
the basis $\big\{\frac1{\sqrt{2\pi}}\e^{in(\,\cdot\,)}\big\}_{n\in\Z}\subset\ltwo(\T)$.
As a result, one obtains that $B_{nn'}(\lambda)=0$, which in turn implies that $Q_k=0$
(see \eqref{def_Bn} and \eqref{def_Q_k}).
\end{Remark}

Second, we consider the case of a general function $V:$

\begin{Lemma}\label{lemma_limit}
Assume that $V\in\linf(\R;\R)$ is $2\pi$-periodic. Then, the remainder term $Q_k$
vanishes asymptotically along the free evolution, that is,
$$
\slim_{t\to\pm\infty}\e^{itH_k^0}\U_kQ_k\U_k^*\e^{-itH_k^0}=0
\quad\hbox{in}\quad\ltwo(\Pi).
$$
\end{Lemma}

\begin{proof}
The equations \eqref{leadingterm}-\eqref{remainderterm}, Proposition \ref{prop_leading}
and the results of this section imply that
$$
\big\langle\U_k\big(W_{k,-}-1\big)\U_k^*\xi,\zeta\big\rangle_{\Hrond_k}
=\big\langle\U_k\big(1\otimes R(A_+)\big)(S_k-1)\U_k^*\xi,
\zeta\big\rangle_{\Hrond_k}
+\big\langle Q_k\;\!\xi,\zeta\big\rangle_{\Hrond_k},
\quad\xi,\zeta\in\D_k.
$$
Therefore, we deduce from the density of $\D_k$ in $\Hrond_k$ and the unitarity of
$\U_k:\ltwo(\Pi)\to\Hrond_k$ that
\begin{equation}\label{eq_limit_1}
W_{k,-}-1-\big(1\otimes R(A_+)\big)(S_k-1)=\U_k^*Q_k\U_k.
\end{equation}
Also, we know from the existence and completeness of the wave operators
$W_{k,\pm}$ that
\begin{equation}\label{eq_limit_2}
\slim_{t\to\pm\infty}\e^{itH_k^0}S_k\e^{-itH_k^0}=S_k,
\quad
\slim_{t\to\infty}\e^{itH_k^0}W_{k,-}\e^{-itH_k^0}=S_k
\quad\hbox{and}\quad
\slim_{t\to-\infty}\e^{itH_k^0}W_{k,-}\e^{-itH_k^0}=1.
\end{equation}
Furthermore, the definition of the function $R$ (see \eqref{def_R}) and Proposition
\ref{prop_asymp} imply in $\ltwo(\R_+)$ the relations
$$
\slim_{t\to\infty}\e^{it(-\triangle_{\rm N})}R(A_+)\e^{-it(-\triangle_{\rm N})}=1
\quad\hbox{and}\quad
\slim_{t\to-\infty}\e^{it(-\triangle_{\rm N})}R(A_+)\e^{-it(-\triangle_{\rm N})}=0,
$$
which in turn imply in $\ltwo(\Pi)$ the relations
\begin{equation}\label{eq_limit_3}
\slim_{t\to\infty}\e^{itH_k^0}\big(1\otimes R(A_+)\big)\e^{-itH_k^0}=1
\quad\hbox{and}\quad
\slim_{t\to-\infty}\e^{itH_k^0}\big(1\otimes R(A_+)\big)\e^{-itH_k^0}=0.
\end{equation}
Then, one can conclude by combining the equations
\eqref{eq_limit_1}-\eqref{eq_limit_3}.
\end{proof}

\subsection{New formula for the wave operators}

In this final section, we collect the information on the wave operators obtained so
far. The results are stated in two corollaries.

\begin{Corollary}\label{cor_wave_k}
Assume that $V\in\linf(\R;\R)$ is $2\pi$-periodic. Then, we have in $\ltwo(\Pi)$ the
equalities
\begin{equation}\label{eq_W_k_moins}
W_{k,-}-1=\big(1\otimes R(A_+)\big)(S_k-1)+\U_k^*Q_k\U_k
\end{equation}
and
\begin{equation}\label{eq_W_k_plus}
W_{k,+}-1=\big(1-1\otimes R(A_+)\big)(S_k^*-1)+\U_k^*Q_k\U_k S_k^*,
\end{equation}
with $R$ and $Q_k$ given in \eqref{def_R} and \eqref{def_Q_k}. In addition, the term
$Q_k$ satisfies
$$
\slim_{t\to\pm\infty}\e^{itH_k^0}\U_k Q_k\U_k^*\e^{-itH_k^0}=0.
$$
\end{Corollary}

\begin{proof}
As already mentioned in the proof of Lemma \ref{lemma_limit}, the equations
\eqref{leadingterm}-\eqref{remainderterm}, Proposition \ref{prop_leading} and the
results of Section \ref{section_remainder} imply the formula \eqref{eq_W_k_moins} for
$W_{k,-}$. The formula \eqref{eq_W_k_plus} for $W_{k,+}$ follows from
\eqref{eq_W_k_moins} and from the relation $W_{k,+}=W_{k,-}S_k^*$. Finally, the
properties of the term $Q_k$ follow directly from Lemma \ref{lemma_limit}.
\end{proof}

Now, we know from \cite[Sec.~2.4]{Fra03} that the wave operators
$W_\pm\equiv W_\pm(H^0,H^V)$ and the scattering operator $S\equiv S(H^0,H^V)$ for the
pair $\{H^0,H^V\}$ admit direct integral decompositions
$$
\G\;\!W_\pm\;\!\G^{-1}=\int_{[-1/2,1/2]}^\oplus W_{k,\pm}\,\d k
\qquad\hbox{and}\qquad
\G\;\!S\;\!\G^{-1}=\int_{[-1/2,1/2]}^\oplus S_k\,\d k,
$$
with $\G:\ltwo(\R\times\R_+)\to\int_{[-1/2,1/2]}^\oplus\ltwo(\Pi)\,\d k$ the Gelfand
transform of Section \ref{section_direct}. Therefore, one directly infers from
Corollary \ref{cor_wave_k} the following new formulas for $W_\pm:$

\begin{Corollary}\label{cor_full_wave}
Assume that $V\in\linf(\R;\R)$ is $2\pi$-periodic. Then, we have in
$\ltwo(\R\times \R_+)$ the equalities
$$
W_--1=\big(1\otimes R(A_+)\big)(S-1)+Q
\qquad\hbox{and}\qquad
W_+-1=\big(1-1\otimes R(A_+)\big)(S^*-1)+QS^*,
$$
with $Q:=\G^{-1}\left(\int_{[-1/2,1/2]}^\oplus\U_k^*Q_k\U_k\,\d k\right)\G$.
\end{Corollary}

\section{Appendix}\label{Sec_Appendix}
\setcounter{equation}{0}

We present in this appendix a proposition of independent interest on the asymptotic
behaviour of functions of the generator of dilations $A_+$ under the time evolution
generated by the Neumann Laplacian $-\triangle_{\rm N}$. Before this, we recall that
the usual weighted $\ltwo$-spaces are defined by
$$
\H_t(\R):=\left\{\varphi\in\ltwo(\R)
\mid\int_\R\big(1+|x|^2\big)^t|\varphi(x)|^2<\infty\right\},
\quad t\ge0.
$$

\begin{Proposition}\label{prop_asymp}
Let $f\in C^1(\R)$ satisfy $f'\in\H_t(\R)$ for some $t>1/2$ and
$\lim_{x\to\pm\infty}f(x)=f_\pm$ for some $f_\pm\in\C$. Then, one has
\begin{equation}\label{eq_limits}
\slim_{t\to\pm\infty}\e^{it(-\triangle_{\rm N})}f(A_+)\e^{-it(-\triangle_{\rm N})}
=f_\pm.
\end{equation}
\end{Proposition}

\begin{proof}
The operator of multiplication in $\ltwo(\R_+)$ given by
$$
(B\varphi)(x):=\frac12\ln(x^2)\;\!\varphi(x),
\quad\varphi\in C^\infty_{\rm c}(\R_+),
$$
is essentially self-adjoint \cite[Ex.~5.1.15]{Ped89}, with self-adjoint extension
denoted by the same symbol. Also, a direct calculation shows that $B$ and $A_+$
satisfy for $t\in\R$ and $\varphi\in C^\infty_{\rm c}(\R_+)$ the relation
$$
\e^{itB}A_+\e^{-itB}\varphi=(A_+-t)\;\!\varphi.
$$
Since $C^\infty_{\rm c}(\R_+)$ is a core for $A_+$, this implies that
$\e^{itB}A_+\e^{-itB}=(A_+-t)$ as self-adjoint operators. Therefore, one obtains that
$$
\slim_{t\to\pm\infty}\e^{itB}f(A_+)\e^{-itB}
=\slim_{t\to\pm\infty}f\big(\e^{itB}A_+\e^{-itB}\big)
=\slim_{t\to\pm\infty}f(A_+-t)
=f_\mp.
$$
Now, one can apply to the last relation the invariance principle for wave operators
as presented in \cite[Sec.~16.1.1]{BW83} to obtain for each
$\eta\in C^\infty_{\rm c}(\R)$ the relation
\begin{equation}\label{rel_eta}
\slim_{t\to\pm\infty}\e^{it\e^{2B}}f(A_+)\e^{-it\e^{2B}}\eta(B)
=f_\mp\;\!\eta(B).
\end{equation}
For this, one has to check that the function $x\mapsto\e^{2x}$ is
admissible in the sense of \cite[Def.~8.1.16]{BW83} and that the commutator
$Bf(A_+)\eta(B)-f(A_+)\eta(B)B$, defined as a quadratic form on $\dom(B)$,
extends to a trace class operator. The first condition is trivially verified. For the
second condition, we have the following equalities in the form sense on
$C^\infty_{\rm c}(\R_+):$
\begin{align*}
Bf(A_+)\eta(B)-f(A_+)\eta(B)B
&=-i\left(\hbox{s-}\frac\d{\d t}\e^{itB}f(A_+)\e^{-itB}\right)_{t=0}\eta(B)\\
&=-i\left(\hbox{s-}\frac\d{\d t}f(A_+-t)\right)_{t=0}\eta(B)\\
&=if'(A_+)\eta(B).
\end{align*}
Therefore, the commutator $Bf(A_+)\eta(B)-f(A_+)\eta(B)B$ extends to the bounded
operator $if'(A_+)\eta(B)$ by density of $C^\infty_{\rm c}(\R_+)$ in $\dom(B)$.
On another hand, if $\mathcal M:\ltwo(\R_+)\to\ltwo(\R)$ denotes the Mellin
transform as given in \cite[Sec.~1.5]{BE11}, then it is known that
$\mathcal MA_+\mathcal M^{-1}=X$ and $\mathcal MB\mathcal M^{-1}=-P$,
with $X$ the multiplication operator by the variable in $\ltwo(\R)$ and $P$ the
differentiation operator $-i\nabla$ in $\ltwo(\R)$. It follows that
$$
f'(A_+)\eta(B)
=\mathcal M^{-1}f'(X)\eta(-P)\mathcal M,
$$
with $f'(X)\eta(-P)$ of trace class due to the decay assumption on $f'$
(see \cite[Cor.~4.1.4]{ABG96}). Therefore, the operator $f'(A_+)\eta(B)$ is also trace
class, and the second condition is verified. So, the relation \eqref{rel_eta} holds
and implies that
$
\slim_{t\to\pm\infty}\e^{it\e^{2B}}f(A_+)\e^{-it\e^{2B}}\varphi=f_\mp\;\!\varphi
$
for each vector $\varphi\in\ltwo(\R_+)$ such that $\varphi=\eta(B)\varphi$
for some $\eta\in C^\infty_{\rm c}(\R)$. Since this set of vectors $\varphi$
is dense in $\ltwo(\R_+)$, one infers that
$$
\slim_{t\to\pm\infty}\e^{it\e^{2B}}f(A_+)\e^{-it\e^{2B}}=f_\mp.
$$
Finally, using the fact that $\e^{2B}=\Fc\;\!(-\triangle_{\rm N})\Fc^{-1}$ and
$A_+=-\Fc\;\!A_+\Fc^{-1}$ with $\Fc$ the cosine transform \eqref{eq_cosine}, one obtains
from the last relation that
$$
\slim_{t\to\pm\infty}\e^{it(-\triangle_{\rm N})}f(-A_+)\e^{-it(-\triangle_{\rm N})}
=f_\mp,
$$
which is equivalent to \eqref{eq_limits}.
\end{proof}



\end{document}